\documentclass[11pt]{article}

\usepackage[margin=1in]{geometry}

\usepackage[utf8]{inputenc}
\usepackage{microtype}

\usepackage{hyperref}
\hypersetup{colorlinks=true,citecolor=blue,linkcolor=blue,filecolor=blue,urlcolor=blue}

\usepackage{amsmath,amssymb,amsfonts,amsthm,amstext}
\usepackage{dsfont}
\usepackage{enumerate}
\numberwithin{equation}{section}

\usepackage{mathtools}
\mathtoolsset{centercolon}
\makeatletter
\protected\def\tikz@nonactivecolon{\ifmmode\mathrel{\mathop\ordinarycolon}\else:\fi}
\makeatother

\allowdisplaybreaks[3]

\usepackage{cleveref}

\newtheorem{theorem}{Theorem}[section]
\newtheorem{lemma}[theorem]{Lemma}
\newtheorem{proposition}[theorem]{Proposition}

\theoremstyle{definition}
\newtheorem{definition}[theorem]{Definition}
\newtheorem*{remark}{Remark}

\crefname{lemma}{Lemma}{Lemmas}
\crefname{definition}{Definition}{Definitions}
\crefname{theorem}{Theorem}{Theorems}
\crefname{conjecture}{Conjecture}{Conjectures}
\crefname{corollary}{Corollary}{Corollaries}
\crefname{section}{Section}{Sections}
\crefname{appendix}{Appendix}{Appendices}
\crefname{figure}{Fig.}{Figs.}

\newcommand{\be}{\begin{align}}
\newcommand{\ee}{\end{align}}

\newcommand\numberthis{\addtocounter{equation}{1}\tag{\theequation}}

\newcommand{\cB}{\mathcal{B}}

\newcommand{\cD}{\mathcal{D}}

\newcommand{\cF}{\mathcal{F}}

\newcommand{\cH}{\mathcal{H}}

\newcommand{\cJ}{\mathcal{J}}
\newcommand{\cK}{\mathcal{K}}

\newcommand{\cN}{\mathcal{N}}
\newcommand{\cO}{O}
\newcommand{\cP}{\mathcal{P}}

\newcommand{\cS}{\mathcal{S}}
\newcommand{\cT}{\mathcal{T}}
\newcommand{\cU}{\mathcal{U}}

\newcommand{\cW}{\mathcal{W}}

\newcommand{\trho}{\tilde{\rho}}

\newcommand{\brho}{\bar{\rho}}

\DeclareMathOperator{\tr}{tr}

\DeclareMathOperator{\id}{id}

\DeclareMathOperator{\supp}{supp}
\DeclareMathOperator*{\argmax}{arg\,max}
\newcommand{\ket}[1]{\left|#1\right\rangle}
\newcommand{\bra}[1]{\left\langle #1\right|}

\newcommand{\Be}[1]{\mathcal{B}_#1}

\newcommand{\Dmax}{D_\text{\normalfont max}}

\newcommand{\Hmax}{H_\text{\normalfont max}}

\newcommand{\Hmin}{H_\text{\normalfont min}}

\newcommand{\Qet}[2]{Q_{\text{\normalfont et}}^{#1,\,#2}}

\newcommand{\Fent}{F_{\text{\normalfont ent}}}

\newcommand{\Favg}{F_{\text{\normalfont avg}}}
\newcommand{\cFavg}{\cF_{\text{\normalfont avg}}}
\newcommand{\Q}[2]{Q^{#1,\,#2}}

\newcommand{\Ceps}[2]{C^{#1,\,#2}}

\newcommand{\sumi}{\sum\nolimits}
\newcommand{\invP}[1]{\Phi^{-1}\left(#1\right)}

\newcommand{\diff}{\text{\normalfont d}}

\newcommand{\ox}{\otimes}
\newcommand{\n}{{\ox n}}

\newcommand{\one}{I}
\newcommand{\eps}{\varepsilon}

\newcommand{\qqquad}{\qquad\qquad\qquad}

\newcommand{\sqtrho}{\sqrt{\tilde{\rho}}}
\newcommand{\dd}{\text{\rm{d}}}

\begin{document}

\title{Decoding quantum information via the Petz recovery map}

\author{Salman Beigi\thanks{
School of Mathematics, Institute for Research in Fundamental Sciences (IPM), Tehran, Iran}
\and Nilanjana Datta\thanks{Statistical Laboratory, Centre for Mathematical
Sciences, University of Cambridge, Wilberforce Road, Cambridge CB3 0WB, UK}
\and Felix Leditzky\footnotemark[2]}

\date{November 2, 2019}

\maketitle

\begin{abstract}
We obtain a lower bound on the maximum number of qubits,  $\Q{n}{\eps}(\cN)$, which can be transmitted over $n$ uses of a quantum channel $\cN$, for a given non-zero error threshold $\eps$. 
To obtain our result, we first derive a bound on the one-shot entanglement transmission capacity of the channel, and then compute its asymptotic expansion up to the second order. In our method to prove this achievability bound, the decoding map, used by the receiver on the output of the channel, is chosen to be the {\em{Petz recovery map}} (also known as the {\em{transpose channel}}). Our result, in particular, shows that this choice of the decoder can be used to establish the coherent information as an achievable rate for quantum information transmission. Applying our achievability bound to the 50-50 erasure channel (which has zero quantum capacity), we find that there is a sharp error threshold above which $\Q{n}{\eps}(\cN)$ scales as $\sqrt{n}$. 
\end{abstract}


\section{Introduction}\label{sec:intro}
The capacity $Q(\cN)$ of a quantum channel $\cN$, for the transmission of quantum information, is referred to as its quantum capacity, and is given by the regularized coherent information of the channel~\cite{Llo97,Sho02,Dev05}:
\begin{align}\label{reg-cap}
Q(\cN) = \lim_{n \to \infty} \frac{1}{n} I_c(\cN^{\otimes n}),
\end{align}
where $I_c(\cN) \coloneqq  \max_\rho I_c(\cN, \rho)$ is the maximum coherent information of the channel over all input states $\rho$ (cf.~\Cref{def:channel-coherent-information}).
The regularization in~\eqref{reg-cap} is necessary in general. However, for the class of {\em{degradable quantum channels}}~\cite{DS05} the coherent information is additive, and hence, for such channels the quantum capacity is given by the single letter formula $Q(\cN)=I_c(\cN)$.

The expression (\ref{reg-cap}) for the quantum capacity is obtained in the so-called {\em{asymptotic, memoryless setting}}, that is, in the limit of infinitely many uses of the channel (which is assumed to be memoryless, i.e., it is assumed that there is no correlation in the noise acting on successive inputs to the channel), under the requirement that the error incurred in the protocol vanishes in the asymptotic limit. It is, however, unrealistic to assume that a quantum channel is used infinitely many times. Instead, it is more meaningful to consider a finite number (say $n$) of uses of the channel and to study the trade-off between the error and the optimal rate of information transmission.

Let $\Q{n}{\eps}(\cN)$ be the maximum number of qubits that can be sent through $n$ uses of a memoryless quantum channel $\cN$, such that the error incurred in the transmission is at most $\eps$ (see \cref{sec:protocols} for precise definitions). We are interested in the behavior of $\Q{n}{\eps}(\cN)$ for large but finite $n$, as a function
of $\eps$. From the above discussion it is evident that $\Q{n}{\eps_n}(\cN)=nQ(\cN) + o(n)$ for an appropriate sequence $\{\eps_n\}_{n\geq 1}$ that vanishes in the asymptotic limit. However, our aim is to find a more refined expansion for $\Q{n}{\eps}(\cN)$. 

\subsection{Main result}
Our main result in this paper is proving a lower bound  on $\Q{n}{\eps}(\cN)$ of the form
\begin{align}\label{eq:SOE-Petz}
\Q{n}{\eps}(\cN)\geq nI_c(\cN) + \sqrt{nV_\eps(\cN)}\, \Phi^{-1}(\eps) + \cO(\log n).
\end{align}
Here $V_\eps(\cN)$ is an $\eps$-dependent characteristic of the channel $\cN$ which we call the $\eps$-\emph{quantum dispersion}; it takes one of two values for the ranges $\eps\in (0,1/2)$ and $\eps\in(1/2,1)$, respectively, as defined in \eqref{v-eps} in \Cref{sec:entropies}. 
Moreover, $\Phi^{-1}(\eps)$ denotes the inverse of the cumulative distribution function of the standard normal distribution (cf.~\Cref{sec:notation}).
We refer to a bound of the form \eqref{eq:SOE-Petz} as a {\em{second order achievability bound}}.

The significance of our proof of~\eqref{eq:SOE-Petz} is that at the decoder we use the \emph{Petz recovery map}~\cite{Petz88}, also known as the \emph{transpose channel}. In previous proofs of achievable bounds on quantum capacity using the \emph{decoupling theorem} \cite{HHWY08,DBWR14}, the decoding map is only given implicitly, its existence being guaranteed by Uhlmann's Theorem.  In our proof however, the decoding map is introduced explicitly and depends solely on the code space. We will explain our method in some detail in \Cref{sec:overview}.

We also examine the bound (\ref{eq:SOE-Petz}) for the 50-50 (symmetric) quantum erasure channel. This is a quantum channel which, with equal probability, transmits the input state undistorted, or replaces it with an ``erasure state'', the latter being a fixed pure state in the orthocomplement of the input Hilbert space of the channel. The capacity of this channel is known to be zero~\cite{BDS97} by the No-cloning theorem, i.e., $Q(\cN)=0$ for the 50-50 quantum erasure channel $\cN$. A stronger result that can be proved based on ideas from~\cite{MW14} is that for any error  $0<\eps< 1/{2}$ we have  $\Q{n}{\eps}(\cN)=O(1)$. (Note that the error criterion used in~\cite{MW14} is different from ours, so the error threshold there is $1/\sqrt{2}$.)
However, our bound (\ref{eq:SOE-Petz}) implies that $\Q{n}{\eps}(\cN) \geq \Omega(\sqrt n)$ for any $\eps> 1/2$. We note that this bound cannot be obtained from the previous one-shot bounds for quantum capacity~\cite{BD10a, DH11}.

\subsection{Related works}

Our lower bound (\ref{eq:SOE-Petz}) is reminiscent of the second order asymptotic expansion for the maximum number
of bits of information which can be transmitted through $n$ uses of a discrete, memoryless
classical channel $\cW$, with an average probability of error of at most $\eps$ denoted by $\Ceps{n}{\eps}(\cW)$.
Such an expansion was first derived by Strassen in 1962~\cite{Str62} and refined by Hayashi~\cite{Hay09} as well as
Polyanskiy, Poor and Verd\'u~\cite{PPV10}. It is given by
\begin{align}\label{c-soa}
\Ceps{n}{\eps}(\cW)= nC(\cW) + \sqrt{nV_\eps(\cW)}\, \Phi^{-1}(\eps) + \cO(\log n),\end{align}
where $C(\cW)$ denotes the capacity of the channel (given by Shannon's formula~\cite{Sha48}) and $V_\eps(\cW)$ is an $\eps$-dependent characteristic of the channel called its $\eps$-{\em{dispersion}}~\cite{PPV10}.

In the last decade there has been a renewal of interest in the evaluation of
second order asymptotics for other classical information-theoretic tasks (see e.g.~\cite{Hay08,Hay09,KH13}
and references therein) and, more recently, even in third-order asymptotics~\cite{KV13a}. The study of
second order asymptotics in Quantum Information Theory was initiated by Tomamichel and Hayashi~\cite{TH13} and Li~\cite{Li14}. The achievability parts of
the second order asymptotics for the tasks studied in~\cite{TH13,Li14} were later also obtained in~\cite{BG13} via the collision relative entropy.

Our second order achievability bound~\eqref{eq:SOE-Petz} is similar in form to~\eqref{c-soa}. Nevertheless, its optimality is open. 
Note that it follows from the strong converse property of the quantum capacity of generalized dephasing channels~\cite{TWW14} that, for such channels, $I_c(\cN)$ is exactly equal to the first order asymptotic
rate (and not just a lower bound on it) for any $\eps\in (0,1)$. Moreover, from the result of~\cite{MW14} it follows that, for degradable channels, the first order asymptotic rate is given by $I_c(\cN)$ for 
$\eps\in(0,1/2)$. 
For the qubit dephasing channel $\mathcal{Z}_\gamma$, Tomamichel et al.~\cite[eq.~(6)]{TBR15} proved a third order asymptotic expansion of the \emph{$\eps$-error $n$-blocklength entanglement transmission capacity} $\Qet{n}{\eps}(\mathcal{Z}_\gamma)$ (see Section \ref{sec:ent-trans} for a definition) that matches our lower bound \eqref{eq:SOE-Petz} up to third and higher order terms.
	Note that for any quantum channel $\cN$ the $n$-blocklength capacities $\Q{n}{\eps}(\cN)$ and $\Qet{n}{\eps}(\cN)$ regularize to the same quantity, the quantum capacity $Q(\cN)$, in the limits $n\to\infty$ and $\eps\to 0$ (see Section \ref{sec:ent-trans} for details).
Moreover, for yet another one-shot variant of the quantum capacity called the \emph{entanglement generation capacity}, a second order asymptotic lower bound similar to \eqref{eq:SOE-Petz} can be obtained using the one-shot achievability bound of~\cite[Proposition 20]{MW14}. Nevertheless, the proof of this one-shot  bound is via the \emph{decoupling theorem}~\cite{HHWY08,DBWR14}. As mentioned above, in contrast to our method, the decoupling theorem does not explicitly provide a decoder.

The Petz recovery map (or transpose channel) was introduced by Petz~\cite{Petz86, Petz88} (see also~\cite{OP93}). In~\cite{BK02} it was shown that, if the Petz recovery map is used as the decoding operation, then the average error incurred in sending an ensemble of commuting states through a quantum channel is at most twice the minimum error.
Later, this map was also used to characterize so-called quantum Markov chain states~\cite{HJPW04}. Furthermore, it was used to study quantum error correcting codes in~\cite{HKM10}. 

Our work should be considered as a new step towards understanding the usefulness of the Petz recovery map. 
In particular it would be interesting to see whether the ideas in our work can be used to show tight achievability bounds for other quantum protocols such as quantum state merging and quantum state redistribution. Another open question in this area is the optimality of the Petz recovery map in the Fawzi-Renner inequality \cite{FR14} for approximate Markov chain states (see \cite{JRS+15} and references therein for a discussion of this question).

\subsection{Overview of our method: Decoding via the Petz recovery map}\label{sec:overview}

Following a standard procedure, our strategy for obtaining the above lower bound \eqref{eq:SOE-Petz} is to first prove a so-called \emph{one-shot} lower bound, i.e., a lower bound on $\Q{1}{\eps}(\cN)$ in terms of the \emph{information spectrum relative entropy} \cite{TH13} (cf.~\Cref{def:inf-spec-rel-entropy}). Then $\Q{n}{\eps}(\cN)$ can be estimated using this one-shot bound applied to the channel $\cN^{\otimes n}$. Since $\Q{n}{\eps}(\cN)=\Q{1}{\eps}(\cN^{\otimes n})$ by definition, we arrive at \eqref{eq:SOE-Petz} by applying the results of~\cite{TH13} for computing the second order asymptotic expansion of the information spectrum relative entropy (see \Cref{prop:entropies-SOA}).

To prove the lower bound~\eqref{eq:SOE-Petz} on $\Q{n}{\eps}(\cN)$, we also study a related information-processing task called entanglement transmission, which entails the transmission of entanglement from the sender to the receiver over multiple uses of the channel $\cN$. It is known~\cite{BKN00,KW04,Dev05} that in the asymptotic, memoryless setting, as well as in the one-shot setting~\cite{KW04,BD10}, quantum capacity and entanglement transmission capacity are related.  
So we prove the lower bound of equation~\eqref{eq:SOE-Petz} for the entanglement transmission capacity and then relate it to the information transmission capacity. This result is stated as \Cref{thm:et-second-order} in \Cref{sec:et-asymptotic}. 

In our coding theorem, as usual, we choose the code subspace randomly, yet fixing the decoding map to be given by the Petz recovery map. To be more precise, let $\cN_{A\rightarrow B}$ be a quantum channel. We choose some positive semidefinite matrix $S_A$ at random, and define the code space as the support of $S$. We then fix the decoder to be 
\begin{align}\label{eq:decoder-32}
\cD = \Gamma_S^{\frac{1}{2}} \circ \cN^*\circ \Gamma_{\cN(S)}^{-\frac{1}{2}}.
\end{align}
Here $\cN^*_{B\rightarrow A}$ is the adjoint map of $\cN$ (defined with respect to the Hilbert-Schmidt inner product), and the map $\Gamma_X$, for an arbitrary $X$ is defined by $\Gamma_X(\rho) = X\rho X^\dagger$. We note that the CPTP map $\cD$ maps any quantum state $\sigma_B$ to a state inside the code space $\supp S_A$, and hence is a valid decoder. The main result of this paper is that with this decoder we can achieve the bound~\eqref{eq:SOE-Petz} on the quantum capacity.  

The technical part of our proof consists of the following steps. We first observe that the \emph{average fidelity}, when the decoder is chosen as in~\eqref{eq:decoder-32}, can be written in terms of the so-called collision relative entropy (cf.~\Cref{def:collision-entropy}). We then obtain a lower bound on the expected value of the average fidelity over the random choice of the code space. This lower bound resembles the coherent information of the channel, and should be considered as our main technical contribution. 

Our method is a generalization of the one introduced in~\cite{BG13} in which the pretty good measurement (or square root measurement) was used as the decoding measurement for the (classical) capacity of classical-quantum channels. In~\cite{BG13} it was observed that the probability of correct decoding using the pretty good measurement can be written in terms of the collision relative entropy. Then the joint convexity of the exponential of collision relative entropy was used to obtain a lower bound on the expected probability of correct decoding over the random choice of the codebook. 

Here, for entanglement transmission, we follow similar ideas: We replace the pretty good measurement with the Petz recovery map, and the probability of correct decoding with the average fidelity. However, the joint convexity of the exponential of the collision relative entropy is not enough to obtain the desired lower bound. We overcome this difficulty by proving a weak monotonicity property of the collision relative entropy under dephasing (\Cref{lem:key-lemma}), which together with the joint convexity gives the final result. This lemma and its proof technique may be of independent interest. 

\subsection{Organization of the paper}
In the following section we fix our notation, introduce the required definitions, and collect some basic tools that are needed to 
prove our main results. In particular, the definitions of the relevant entropic quantities are given in \Cref{sec:entropies}. 
\Cref{sec:ent-trans} contains the main results of this paper in which we prove a one-shot achievability bound and a second order achievability bound. 
In \Cref{sec:example} we discuss the implications of our results for the special case of the 50-50 (symmetric) erasure channel. 
We defer some of the proofs of results in the main text to the appendices.

\section{Preliminaries} \label{sec:prelim}

\subsection{Notation}\label{sec:notation}
Let $\cB(\cH)$ denote the algebra of linear operators acting on 
a Hilbert space $\cH$. In the following we only consider finite-dimensional Hilbert spaces. We denote the set of positive semidefinite operators by $\cP(\cH)$. Let $\cD_{\leq}(\cH)\coloneqq  \{ \rho \in \cP(\cH)\mid \tr \rho \leq 1\}$ be the  set of sub-normalized quantum states, and $\cD(\cH)$ 
be the set of normalized states (density matrices). For a pure state $|\psi\rangle\in\cH$, we use the abbreviation $\psi\equiv |\psi\rangle\langle\psi| \in \cD(\cH)$ for the corresponding density matrix. For $A\in \cP(\cH)$, we write $\supp(A)$ for the support of $A$, i.e., the span of eigenvectors of $A$ corresponding to \emph{positive} eigenvalues. Moreover, we let $\Pi_A$ be the orthogonal projection onto $\supp(A)$. 
For Hermitian $X, Y\in \cB(\cH)$ we let $\lbrace X\leq Y\rbrace$ denote the orthogonal projection on the span of eigenvectors of $X-Y$ with \emph{non-negative} eigenvalues.
Hilbert spaces are often indexed by uppercase letters as in $\cH_A$, and for simplicity of notation we denote $\cH_A\otimes \cH_B$ by $\cH_{AB}$.
Operators acting on $\cH_A$ are distinguished by subscripts as in $X_A\in \cB(\cH_A)$.  The identity operator in $\cB(\cH_A)$ is denoted by $\one_A$, and we often omit it by using the shorthand $X_BY_{AB}\equiv (\one_A\ox X_B)Y_{AB}$.  The completely mixed state is denoted by $\pi_A=I_A/d$ where $d=\dim \cH_A$.

A quantum channel is a linear, completely positive, trace preserving (CPTP) map $\Lambda\colon \cB(\cH_A)\rightarrow\cB(\cH_B)$, which we denote by $\Lambda\colon A\rightarrow B$ or $\Lambda_{A\rightarrow B}$. Given a Stinespring isometry $\cU_\Lambda\colon \cH_A\rightarrow \cH_{BE}$ of a quantum channel $\Lambda_{A\rightarrow B}$ such that $\Lambda(\rho_A) = \tr_E(\cU_\Lambda\rho_A\cU_\Lambda^\dagger)$, we define the complementary channel $\Lambda^c\colon A\rightarrow E$ as $\Lambda^c(\rho_A) = \tr_B(\cU_\Lambda\rho_A \cU_\Lambda^\dagger)$. The identity channel acting  on $\cB(\cH_A)$ is denoted by $\id_A$.
The algebra $\cB(\cH)$ is equipped with the Hilbert-Schmidt inner product: for $X, Y\in \cB(\cH)$ it is defined as 
$\langle X, Y\rangle = \tr(X^{\dagger}Y)$ where $X^{\dagger}$ is the adjoint of $X$. Then for a quantum channel $\Lambda\colon A\rightarrow B$ we may consider its adjoint map $\Lambda^{*}\colon B\rightarrow A$ determined by $\langle X_A, \Lambda^{*}(Y_B)\rangle  = \langle \Lambda(X_A), Y_B\rangle$. Let $\cH_A\cong\cH_B$ be isomorphic Hilbert spaces with $\dim\cH_A=\dim\cH_B=d$ and fix bases $\lbrace |i_A\rangle\rbrace_{i=1}^d$ and $\lbrace |i_B\rangle\rbrace_{i=1}^d$ in them. Then we define a maximally entangled state (MES) of Schmidt rank $m$ to be
\begin{align}\label{eq:MES}
|\Phi_{AB}^m\rangle \coloneqq  \frac{1}{\sqrt{m}} \sum_{i=1}^m |i_A\rangle\ox |i_B\rangle \in\cH_M\ox\cH_{M'},
\end{align}
where $\cH_M\subseteq \cH_A$, $\cH_{M'}\subseteq \cH_B$, and $\cH_M\cong\cH_{M'}$ are isomorphic subspaces with bases $\lbrace |i_A\rangle\rbrace_{i=1}^m$ and $\lbrace |i_B\rangle\rbrace_{i=1}^m$. If $m=d$, we abbreviate $|\Phi_{AB}\rangle\equiv |\Phi_{AB}^d\rangle$. We define the Choi state $\tau_{A'B}=\cJ(\Lambda)$  of a quantum channel $\Lambda_{A\rightarrow B}$ by 
$\tau_{A'B} \coloneqq  (\id_{A'}\ox \Lambda)(\Phi_{A'A}),$
where $\cH_{A'}\cong\cH_A$.

The inverse of the cumulative distribution function of a standard normal random variable is defined by
$\Phi^{-1}(\eps) \coloneqq  \sup\lbrace z\in\mathbb{R}\mid\Phi(z)\leq \eps\rbrace,$
where $\Phi(z) = \frac{1}{\sqrt{2\pi}}\int_{-\infty}^z e^{-t^2/2}dt$. Since $\Phi^{-1}$ is continuously differentiable, we have the following lemma:
\begin{lemma}[{\cite{DL14b}}]\label{lem:phi-trick}
	Let $\eps>0$, then $$\sqrt{n}\,\Phi^{-1}\left(\eps\pm\frac{1}{\sqrt{n}}\right) = \sqrt{n}\,\Phi^{-1}(\eps)\pm \left(\Phi^{-1}\right)'(\xi)$$ for some $\xi$ with $|\xi-\eps|\leq\frac{1}{\sqrt{n}}$.
\end{lemma}

\subsection{Distance measures and entropic quantities}\label{sec:entropies}
In this subsection we collect definitions of useful distances measures and entropic quantities that we employ in our proofs.
We refer to \Cref{sec:props-entropies} for a list of properties of these quantities.
\begin{definition}[Fidelities]\label{def:distances}
Let $\rho,\sigma\in\cD(\cH)$.
\begin{enumerate}[(i)]
\item The \emph{fidelity} $F(\rho,\sigma)$ between $\rho$ and $\sigma$ is defined as
$$F(\rho,\sigma) \coloneqq  \|\sqrt{\rho}\sqrt{\sigma}\|_1=\tr\sqrt{\sqrt{\rho}\sigma\sqrt{\rho}}.$$ For pure states $\psi,\varphi\in\cH$, the fidelity reduces to $F(\psi,\varphi) = |\langle \psi|\varphi\rangle|.$ We use this definition of fidelity even if one of $\rho$ or $\sigma$ is sub-normalized. We also use the notation $F^2(\rho,\sigma)\equiv (F(\rho,\sigma))^2$.
\item \cite{Nie02} The \emph{average fidelity} of a quantum operation $\Lambda$ acting on $\cB(\cH)$ is defined by
\begin{align*}
\Favg(\Lambda;\cH) \coloneqq \int_{\phi\in\cH} \diff\mu(\phi) \bra{\phi} \Lambda(\phi) \ket{\phi},
\end{align*}
where $\diff\mu(\phi)$ is the uniform normalized (Haar) measure on unit vectors $\ket \phi\in\cH$.

\end{enumerate}
\end{definition}

\begin{definition}[Relative entropy and quantum information variance]
Let $\rho\in\cD(\cH)$ and $\sigma\in\cP(\cH)$ be such that $\supp\rho\subseteq\supp\sigma$. Then the \emph{quantum relative entropy} is defined as $$D(\rho\|\sigma)\coloneqq  \tr[\rho(\log\rho-\log\sigma)],$$
and the \emph{quantum information variance} is defined as 
\begin{align*}
V(\rho\|\sigma)\coloneqq  \tr\left[\rho(\log\rho-\log\sigma)^2\right]-D(\rho\|\sigma)^2.
\end{align*}
\label{def:rel-ent-inf-var}
\end{definition}
The von Neumann entropy of a state $\rho\in\cD(\cH)$ is given by $H(\rho)\coloneqq -\tr (\rho\log\rho) = -D(\rho\|\one)$, and we write $H(A)\equiv H(\rho_A)$.
For a bipartite state $\rho_{AB}\in\cD(\cH_A\ox\cH_B)$, the conditional entropy $H(A|B)_\rho$ and mutual information $I(A; B)_\rho$ are defined as 
$H(A|B)_\rho \coloneqq H(AB)_\rho - H(B)_\rho = -D(\rho_{AB}\|\one_A\ox\rho_B)$  and $I(A; B)_\rho \coloneqq H(A)_\rho - H(A|B)_\rho = D(\rho_{AB}\|\rho_A\ox\rho_B)$, respectively.

The following entropic quantities play a key role in our proofs. 
\begin{definition}[Max-relative entropy; {\cite{Dat09,Ren05}}]\label{def:max-relative-entropy}
 Let $\rho\in\cD(\cH)$ and $\sigma\in\cP(\cH)$, then the \emph{max-relative entropy} $\Dmax(\rho\|\sigma)$ is defined as
\begin{align*}
\Dmax(\rho\|\sigma) \coloneqq  \inf \left\lbrace \lambda\in\mathbb{R} \mid \rho\leq 2^\lambda\sigma\right\rbrace.
\end{align*}
For $\eps\in(0,1)$, the \emph{smooth max-relative entropy} $\Dmax^\eps(\rho\|\sigma)$ is defined as
\begin{align*}
\Dmax^\eps(\rho\|\sigma) \coloneqq  \min_{\brho\in\Be{\eps}(\rho)}\Dmax(\brho\|\sigma).
\end{align*}
where $\Be{\eps}(\rho)\coloneqq  \lbrace\trho\in\cD_\leq(\cH)\mid F^2(\rho,\trho)\geq 1-\eps^2\rbrace$
is a ball of \emph{sub-normalized} states around the (normalized) state $\rho$.
\end{definition}

\begin{definition}[Conditional min- and max-entropy; {\cite{Ren05,Tom12}}]\label{def:min-entropy}
 Let $\rho_{AB}\in\cD_\leq(\cH_{AB})$, then the \emph{conditional min-entropy} $\Hmin(A|B)_\rho$ is defined as 
\begin{align*}
\Hmin(A|B)_\rho \coloneqq  -\min_{\sigma_B\in\cD(\cH_B)} \Dmax(\rho_{AB}\|\one_A\ox \sigma_B).
\end{align*}
For $\eps\in(0,1)$, the \emph{smooth conditional min-entropy} is defined as
\begin{align*}
\Hmin^\eps(A|B)_\rho \coloneqq  \max_{\brho_{AB}\in\Be{\eps}(\rho_{AB})} \Hmin(A|B)_{\brho},
\end{align*}
and the \emph{smooth conditional max-entropy} is defined as
\begin{align*}
\Hmax^\eps(A|B)_\rho \coloneqq \min_{\brho_{AB}\in\Be{\eps}(\rho_{AB})} \max_{\sigma_B}\log F^2(\brho_{AB},\one_A\ox\sigma_B).
\end{align*}
\end{definition}

The collision relative entropy is a central quantity in \Cref{sec:ent-trans}. Its conditional version was first defined by Renner \cite{Ren05} in the quantum case.
\begin{definition}[Collision relative entropy]\label{def:collision-entropy}
For $\rho\in\cD(\cH)$ and $\sigma\in\cP(\cH)$ the \emph{collision relative entropy} is defined as
\begin{align*}
D_2(\rho\|\sigma) \coloneqq \log \tr\left(\sigma^{-1/2}\rho\sigma^{-1/2}\rho\right).
\end{align*}
\end{definition}

The second order asymptotic analysis in \Cref{sec:ent-trans} relies on the information spectrum relative entropy, whose quantum version was first introduced in~\cite{TH13}:
\begin{definition}[Information spectrum relative entropy]\label{def:inf-spec-rel-entropy}
Let $\rho\in\cD(\cH)$, $\sigma\in\cP(\cH)$, and $\eps\in (0,1)$. The \emph{information spectrum relative entropy} is defined as
\begin{align*}
D_s^\eps(\rho\|\sigma) \coloneqq \sup\lbrace \gamma\in\mathbb{R}\mid \tr\left(\rho \left\lbrace \rho\leq 2^\gamma \sigma\right\rbrace\right) \leq \eps \rbrace.
\end{align*}
\end{definition}

Finally, we define two important quantities in our discussion, the coherent information $I_c(\cN)$ and the $\eps$-quantum dispersion $V_\eps(\cN)$ of a quantum channel.
\begin{definition}[Coherent information of a quantum channel]\label{def:channel-coherent-information}
Let $\cN\colon A\rightarrow B$ be a quantum channel with Stinespring isometry $\cU_\cN\colon \cH_A\rightarrow \cH_{BE}$. Furthermore, for $\rho_A\in\cD(\cH)$ let $\ket{\Psi_{RA}^\rho}$ be a purification of $\rho_A$ and set 
\begin{align} \label{eq:omega}
\ket{\omega_{RBE}} = (I_R\ox \cU_{\cN})\ket{\Psi_{RA}^\rho}.
\end{align} Then the \emph{coherent information of the quantum channel for the input state $\rho_A$} is defined as
\begin{align*}
I_c(\cN, \rho_A) \coloneqq D(\omega_{RB}\|\one_R\ox \omega_{B}) = - H(R|B)_\omega. 
\end{align*}
whereas the \emph{coherent information of the quantum channel} is defined as
\begin{align}\label{coh-info}
I_c(\cN) \coloneqq \max_{\rho_A} I_c(\cN, \rho_A).
\intertext{Define the set $\cS_c(\cN)$ of quantum states achieving the maximum in \eqref{coh-info} as}
\cS_c(\cN) \coloneqq \argmax_{\rho_A} I_c(\cN, \rho_A),\label{eq:Pi_c}
\end{align}
then we can introduce the \emph{$\eps$-quantum dispersion} $V_\eps(\cN)$ of the quantum channel $\cN$:
\begin{align}\label{v-eps}
V_\eps(\cN) \coloneqq \begin{cases}
\min_{\rho_A\in\cS_c(\cN)} V(\omega_{RB}\|I_R\ox\omega_B) & \text{if }\eps\in (0,1/2)\\
\max_{\rho_A\in\cS_c(\cN)} V(\omega_{RB}\|I_R\ox\omega_B) & \text{if }\eps\in (1/2,1),
\end{cases}
\end{align}
where for each $\rho_A$ the state $\omega_{RBE}$ is defined as in \eqref{eq:omega}.
\end{definition}
\begin{remark}
Since any two purifications of $\rho_A$ are connected by a unitary acting on the purifying system $R$ alone, the coherent information as well as the $\eps$-quantum dispersion are independent of the chosen purification.
\end{remark}


\section{Entanglement transmission via the transpose channel method}\label{sec:ent-trans}

\subsection{Protocols}\label{sec:protocols}
In this section we define the protocols of \emph{quantum information transmission} and \emph{entanglement transmission} and the associated capacities. Suppose Alice and Bob are allowed to communicate via a quantum channel $\cN\colon  \cB(\cH_{A}) \mapsto \cB(\cH_B)$, where 
$\cH_{A}$ denotes the Hilbert space of the system whose state Alice prepares as 
input to the channel, and $\cH_B$ denotes the
Hilbert space of the output system of the channel that Bob receives.

In quantum information transmission, Alice's task is to convey to Bob, over a single use of the channel $\cN$, an arbitrarily chosen pure quantum state, from some Hilbert space $\cH_M\subseteq\cH_A$ of dimension $m$, with an error of at most $\eps$. More precisely, we define a one-shot $\eps$-error quantum code as a triple $(m, \cH_M, \cD)$, where $\cH_M\subseteq \cH_{A}$ is a subspace with $m=\dim \cH_M$ and $\cD: \cB(\cH_B) \rightarrow \cB(\cH_A)$ is a CPTP map such that
\begin{align}\label{f-sq-q}
F_{\text{avg}}(\cN; \cH_M) =  \int_{\ket{\varphi} \in \cH_M} \dd\mu(\varphi)
F^2\left( \varphi, (\cD \circ \cN) (\varphi)\right) \geq 1 - \eps.
\end{align}
Here the embedding $\cH_M\subseteq \cH_A$ is Alice's encoding map and $\cD$ is Bob's decoding map. Note that here the error of the protocol is measured in terms of the \emph{average} fidelity (and not worst case fidelity). Also, note that $\Favg(\cN; \cH_M)$ as defined above, depends on $\cD$, but we suppress it for notational simplicity.   

The {\em{$\eps$-error one-shot quantum capacity}} of the channel $\cN$ is then defined as:
\begin{align}\label{eq:1-qinfo-cap}
\Q{1}{\eps}(\cN) \coloneqq  \sup \{ \log m : \exists (m, \cH_M, \cD) \text{ such that }  F_{\text{avg}}(\cN; \cH_M) \geq 1- \eps\}.
\end{align}
For $n$ successive uses of the memoryless channel, we define the {\em{$\eps$-error $n$-blocklength quantum capacity}} of the channel $\cN$ as
\begin{align}\label{eq:n-qinfo-cap}
\Q{n}{\eps}(\cN) \coloneqq  \Q{1}{\eps}\left(\cN^{\otimes n}\right).
\end{align}

We also consider the entanglement transmission protocol~\cite{DS05} in the one-shot setting. 
In this protocol, Alice's task is to convey to Bob, over a single use of the channel $\cN$, a quantum state on a system $M$ in her possession, such that the entanglement between $M$ and some reference system $R$ that is inaccessible to her is preserved in the process. More precisely, we define a one-shot $\eps$-error entanglement transmission code as a triple $(m, \cH_M, \cD)$, where $\cH_M\subseteq \cH_{A}$ is a subspace with $m=\dim \cH_M$ and $\cD\colon  \cB(\cH_B) \rightarrow \cB(\cH_A)$ is a CPTP map such that
\begin{align}\label{f-sq-ent}
 F_{\text{ent}}(\cN; \cH_M)\coloneqq  F^2(\Phi^m_{RA},(\id_R\otimes \cD\circ\cN)(\Phi^m_{RA})) \geq 1-\eps,
\end{align}
where $\Phi^m_{RA}$ is a maximally entangled state of Schmidt rank $m=\dim \cH_M$ between the subspace $\cH_M\subseteq \cH_A$ and a reference system $R$.  

The {\em{$\eps$-error one-shot entanglement transmission capacity}} 
of the channel $\cN$ is then defined as 
\begin{align}\label{eq:1-et-cap}
\Qet{1}{\eps}(\cN) \coloneqq  \sup \{ \log m : \exists (m, \cH_M, \cD) \text{ such that } F_{\text{ent}}(\cN; \cH_M) \geq 1- \eps\}.
\end{align}
We define
\begin{align}\label{eq:n-et-cap}
\Qet{n}{\eps}(\cN) \coloneqq  \Qet{1}{\eps}\left(\cN^{\otimes n}\right),
\end{align}
and call it the {\em{$\eps$-error $n$-blocklength entanglement transmission capacity}} of the channel $\cN$.

\medskip
$\Q{n}{\eps}(\cN)$ and $\Qet{n}{\eps}(\cN)$ are related quantities. Indeed, by \Cref{lem:ave-ent-fid} we have 
$$F_{\text{avg}}(\cN; \cH_M) = \frac{m F_{\text{ent}}(\cN; \cH_M)+1}{m+1}\geq F_{\text{ent}}(\cN; \cH_M),$$
which means that 
\begin{align}\label{eq:Q-Qet-compare}
\Q{n}{\eps}(\cN) \geq \Qet{n}{\eps}(\cN) .
\end{align}

For the $n$-blocklength capacity $\Q{n}{\eps}(\cN)$ given by \eqref{eq:n-qinfo-cap} above, the corresponding capacity in the asymptotic memoryless setting is defined as
\begin{align}\label{eq:quantum-cap}
Q(\cN) \coloneqq \lim_{\eps\rightarrow 0}\lim_{n\rightarrow \infty} \frac{1}{n}\Q{n}{\eps}(\cN).
\end{align}
Note that the results in \cite{BKN00} and \cite{KW04} show that substituting $\Qet{n}{\eps}(\cdot)$ on the right-hand side of \eqref{eq:quantum-cap} yields the same quantity $Q(\cN)$.
By the Lloyd-Shor-Devetak theorem~\cite{Llo97,Sho02,Dev05} it is given by the following expression:
\begin{align}\label{eq:LSD}
Q(\cN) = \lim_{n\rightarrow \infty} \frac{1}{n} I_c(\cN^{\ox n}).
\end{align}

Note that if $\cN$ is degradable, then by the additivity of coherent information for degradable channels~\cite{DS05} the regularized expression in \eqref{eq:LSD} reduces to the \emph{single letter} expression
\begin{align}\label{eq:quantum-cap-degradable}
Q(\cN) = I_c(\cN).
\end{align}

\subsection{One-shot bound}

The main result of this section is a lower bound on the one-shot $\eps$-error entanglement transmission capacity. We then use this lower bound in \Cref{sec:et-asymptotic} to prove a lower bound on the $n$-blocklength capacities. Throughout this section, we abbreviate $\rho\equiv\rho_A$.

\begin{theorem}[One-shot bound]\label{thm:et-one-shot}
Let $\cN\colon A\rightarrow B$ be a quantum channel and fix $\eps\in(0,1)$. Then for every density matrix $\rho$, the one-shot $\eps$-error entanglement transmission capacity, defined through \eqref{eq:1-et-cap}, satisfies:
\begin{align*}
\Q{1}{\eps}(\cN)\geq\Qet{1}{\eps}(\cN) \geq \min\left\lbrace D_s^{\delta_1}(\id_R\ox \cN(\Psi^\rho)\|I_R\ox \cN(\rho)) + \log \frac{\eps_1-\delta_1}{1-\eps_1},\, D_s^{\delta_2}(\Psi^\rho\| I_R\ox \rho) + \log \frac{\eps_2-\delta_2}{1-\eps_2} \right\rbrace.
\end{align*}
Here, $\Psi^\rho$ is a purification of $\rho$, and for $i=1,2$ the parameters $\eps_i,\delta_i>0$ are chosen such that
\begin{align}\label{eq:def-tilde-eps}
\eps = \left( \eps_1 + \sqrt{\tr(\rho^2)+\eps_2}\right)\left(1+\frac{1}{d}\right),
\end{align}
with $d=\dim\cH_A$ and $0\leq \delta_i\leq \eps_i$.
\end{theorem}

\paragraph{Outline of the proof:}
As explained in the introduction, the proof of this theorem follows from similar ideas as those in the proof of the achievability bound of~\cite{BG13} for the classical capacity of quantum channels. Our proof consists of the following steps:
\begin{enumerate}[(a)]
\item Using random encoding at the sender's side
\item Fixing the decoding map to be the Petz recovery map (transpose channel)
\item Writing the average fidelity in terms of the collision relative entropy
\item Applying weak monotonicity under dephasing (\Cref{lem:key-lemma})
\item Using the joint convexity property of the collision relative entropy (\Cref{lem:collision-convex})
\end{enumerate}

Step (a) is fairly standard and amounts to choosing a random subspace as the code space. We use ideas from~\cite{HSW08} to determine the distribution according to which we pick this random subspace. 

Steps (b) and (c) are borrowed from~\cite{BG13}, where the pretty good measurement was used to decode a classical message. There it was observed that the (average) probability of successful decoding can be written in terms of the collision relative entropy. Here for the entanglement transmission protocol, the pretty good measurement is replaced with the Petz recovery map. Note that in the case of classical-quantum channels the transpose channel reduces to the pretty good measurement; hence, we can regard the former as a generalization of the latter. It is easy to see that the average fidelity for entanglement 
transmission, obtained by the transpose channel method, can be written in terms of the collision relative entropy.

In~\cite{BG13}, in order to prove the achievability bound for the capacity of a classical-quantum channel, the final step was to use the joint convexity property of the collision relative entropy similar to step (e). However, this property by itself is not sufficient for obtaining the desired lower bound. To overcome this problem, in step (d) we prove a key result about a weak monotonicity property of the collision relative entropy under dephasing (\Cref{lem:key-lemma}).
This lemma should be considered as the main new ingredient of our method. After proving it in \Cref{sec:key-lemma} below, we then proceed with the proof of \Cref{thm:et-one-shot} in \Cref{sec:thm-4.1-proof}.

\subsection{Weak monotonicity under dephasing}\label{sec:key-lemma}

In the following, $\cH$ is a Hilbert space of dimension $d$ with the computational orthonormal basis $\{\ket 1, \dots, \ket d\}$. Later we will take $\cH$ to be $\cH_A$, the input space of the channel.  For a vector $\ket \varphi=c_1\ket 1+\cdots + c_d \ket d$ we define 
$$\ket{\varphi^*} \coloneqq c_1^*\ket 1+\cdots + c_d^*\ket d,$$
where $c_i^*$ is the complex conjugate of $c_i$. 

\begin{definition}[Dephasing map]\label{def:dephasing-map}
Let $U$ be a unitary operator acting on $\cH$. Then the vectors $|u_i\rangle\coloneqq  U|i\rangle$ for $i=1,\dots,d$ form an orthonormal basis for $\cH$. We define the \emph{dephasing map} $\cT_U$ associated with $U$ by
\begin{align*}
\cT_U(\rho) \coloneqq  \sum_{i=1}^d |u_i\rangle \langle u_i|\rho|u_i\rangle\langle u_i|.
\end{align*}
\end{definition}
For any unitary operator $U$, define the unitary operator $Z_U$ through the relation 
$$Z_U\ket{u_j} = e^{2\pi {i} j/d} \ket{u_j},$$
where $\lbrace |u_i\rangle\rbrace_{i=1}^d$ is the basis defined above with respect to $U$.
Moreover, for every $X\in \cB(\cH)$ let $\Gamma_X$ be the superoperator defined by 
\begin{align}\label{eq:def-Gamma}
\Gamma_X(\rho)\coloneqq X \rho X^{\dagger}.
\end{align}

The following characterizations of the dephasing map $\cT_U$ follow from straightforward calculations. In the sequel, we abbreviate $u_i \equiv |u_i\rangle\langle u_i|$ and $u_i^* \equiv |u_i^*\rangle\langle u_i^*|$. 

\begin{lemma}\label{lem:dephasing-map}~
\begin{enumerate}[{\normalfont (i)}]
\item
Let $|\Phi\rangle_{RA} = \frac{1}{\sqrt d}\sum_{i=1}^{d} |i\rangle\ox |i\rangle$ be a maximally entangled state. Then for every unitary operator $U$ we have
\begin{align*}
(\cT_U\ox\id_A)(\Phi) = \frac{1}{d}\sum_{i=1}^{d} u_i\ox u_i^*.
\end{align*}

\item For every unitary operator $U$ and $\rho\in\cD(\cH)$, we have
$$\mathcal T_U(\rho) = \frac{1}{d} \sum_{j=0}^{d-1} Z_U^j \rho Z_U^{-j} = \frac 1 d \sum_{j=0}^{d-1} \Gamma_{Z_U}^j(\rho).$$
\end{enumerate}
\end{lemma}

Now we are ready to state and prove the lemma about weak monotonicity under dephasing.
\begin{lemma}[Weak monotonicity under dephasing]
\label{lem:key-lemma}
Let $\Lambda\colon A\rightarrow B$ be a completely positive map and $|\Phi\rangle_{RA} = \frac{1}{\sqrt d}\sum_{i=1}^{d} |i\rangle\ox |i\rangle$ be a maximally entangled state. For a fixed unitary operator $U$, let $\sigma_{RB}$ be a state satisfying
$\Gamma_{Z_U}\ox\id_B(\sigma_{RB}) = \sigma_{RB}$.
Then we have
\begin{align*}
\exp D_2\left(\cT_U\otimes \Lambda(\Phi_{RA})\middle\| \sigma_{RB} \right) \geq \frac{1}{d} \exp D_2\left( \id_R\otimes \Lambda(\Phi_{RA})\middle\| \sigma_{RB}\right).
\end{align*}
\end{lemma}
\begin{proof}
We use the following notation: Fix a unitary operator $U$, and for $0\leq j\leq d-1$ define
\begin{align*}
\rho^{(j)}_{RB} := \Gamma_{Z_U}^j\otimes \Lambda(\Phi_{RA}) \qquad \text{and}\qquad \rho_{RB} := \frac 1 d\sum_{j=0}^{d-1} \rho_{RB}^{(j)},
\end{align*}
By \Cref{lem:dephasing-map} we have $\rho_{RB} = \cT_U\otimes \Lambda(\Phi_{RA}),$ 
which implies that
\begin{align}\label{eq:rho-invariance}
\Gamma_{Z_U}\otimes \id_B( \rho_{RB}) =  \rho_{RB}.
\end{align}
We then compute:
\begin{align*}
\exp D_2 \left( \mathcal T_U\otimes \Lambda(\Phi) \| \sigma_{RB}\right) &=  \exp D_2\left( \rho_{RB} \|  \sigma_{RB}\right) \\
& = \tr\left[   \rho  \sigma^{-1/2}  \rho  \sigma^{-1/2}   \right]\\
& = \frac{1}{d} \sum_{j=0}^{d-1} \tr\left[ \Gamma_{Z_U}^j \otimes \Lambda (\Phi)\cdot  \sigma^{-1/2}  \rho  \sigma^{-1/2}   \right]\\
& = \frac{1}{d} \sum_{j=0}^{d-1} \tr\left[ \id_R \otimes \Lambda (\Phi)\cdot \Gamma_{Z_U}^{-j} \otimes\id_B \left( \sigma^{-1/2}  \rho  \sigma^{-1/2}\right)   \right]\\
& =  \tr\left[ \id_R \otimes \Lambda (\Phi)\cdot \left( \sigma^{-1/2}  \rho  \sigma^{-1/2}\right)   \right]\\
&= \frac{1}{d} \sum_{j=0}^{d-1}\tr\left[ \id_R \otimes \Lambda (\Phi)\cdot \left( \sigma^{-1/2}  \rho^{(j)}  \sigma^{-1/2}\right)   \right]\\
& \geq  \frac{1}{d} \tr\left[ \id_R \otimes \Lambda (\Phi)\cdot \left( \sigma^{-1/2}  \rho^{(0)}  \sigma^{-1/2}\right)   \right]\\
& = \frac{1}{d} \exp D_2\left(\id_R \otimes \Lambda (\Phi) \middle\|   \sigma  \right),
\end{align*}
where the fifth equality follows from \eqref{eq:rho-invariance} as well as the assumption that $\Gamma_{Z_U}\ox\id_B(\sigma) = \sigma$, and the inequality follows from the fact that both the operators $\id_R\otimes \Lambda(\Phi)$ and $ \sigma^{-1/2}  \rho^{(j)}  \sigma^{-1/2}$ (for all $j$) are positive semidefinite.
\end{proof}

\subsection{Proof of \texorpdfstring{\Cref{thm:et-one-shot}}{Theorem \ref{thm:et-one-shot}}}\label{sec:thm-4.1-proof}
Let $\rho_A$ be an input state of the channel $\cN\colon {A\rightarrow B}$. We assume without loss of generality that $\rho$ is full-rank, since otherwise we may restrict $\mathcal H_A$ to the support of $\rho$. In the following we use the notation $\tilde \rho = d \rho.$ Furthermore, we let 
$$\Psi^{\rho}_{RA} = (I_R\otimes \sqrt{\trho})\Phi_{RA}(I_R\otimes  \sqrt{\trho})$$ 
be a purification of $\rho_A$.

\paragraph{Code construction:} Let $m$ be a positive integer to be determined. Let $P$ be a rank-$m$ projection acting on $\mathcal H_A$. Later we will assume that $P$ is chosen randomly according to the Haar measure, but for now we assume that $P$ is fixed. 
Define 
$$S\coloneqq  \sqrt{\trho} P \sqrt{\trho},$$
and let $\Pi_S$ be the projection onto the support of $S$, i.e.,
$$\Pi_S=S^{-1/2} \sqtrho P\sqtrho S^{-1/2}.$$ 
We choose $\supp(S)$ as the code space. 
Since $\rho$ is full-rank, $S$ has rank $m$. As a result, the code space $\supp(S)$ is of dimension $m$.
 
With the above construction, for every $\ket\psi\in \supp(P)$ we have $\sqrt{\trho}\ket\psi \in \supp(S)$. Moreover, $S^{-1/2}\sqtrho\ket \psi \in \supp(S)$ where $S^{-1}$ is computed on its support. 

Let us fix some orthonormal basis $\{\ket{v_1}, \dots, \ket{v_m}\}$ of $\supp(P)$, and define 
$$\ket{w_i}\coloneqq S^{-1/2} \sqtrho \ket{v_i}.$$
Then for every $i$ we have $\ket{w_i}\in \supp(S)$. Furthermore, 
$$\sum_{i=1}^m \ket{w_i}\bra{w_i} = \sum_{i=1}^m S^{-1/2} \sqtrho \ket{v_i}\bra{v_i} \sqtrho S^{-1/2} = S^{-1/2} \sqtrho P \sqtrho S^{-1/2} =\Pi_S.$$
This implies that $\{\ket{w_1}, \dots, \ket{w_m}\}$ is an orthonormal basis of the code space. In particular, for every 
$\ket\psi\in \supp(P)$ we have
\begin{align}
\| S^{-1/2} \sqtrho \ket{\psi}  \|=1. \label{eq:normalization}
\end{align}

\paragraph{Decoder:}
For the decoder we choose the Petz recovery map (or transpose channel):
\begin{align}\label{eq:petz-recovery-map}
\mathcal D: = \Gamma_S^{1/2} \circ \cN^*\circ \Gamma_{\cN(S)}^{-1/2}.
\end{align}
Recall that $\Gamma_X$ is defined in \eqref{eq:def-Gamma} and $\cN^*\colon B\rightarrow A$ is the adjoint of $\cN$.
It is easy to verify that $\mathcal D\colon B\rightarrow A$ is a CPTP map, and hence a valid decoder. It satisfies the property $\mathcal D(\cN(S))=S$.

By the definition of the adjoint map, for every $\ket\phi\in \supp(S)$ we have 
\begin{align*}
F^2(\phi, \mathcal D\circ \cN(\phi))& = \tr\left(\phi \cdot \Gamma_S^{1/2} \circ \cN^*\circ \Gamma_{\cN(S)}^{-1/2} \circ N(\phi)\right)\\
& = \tr \left(\cN\circ\Gamma_S^{1/2}(\phi) \cdot \Gamma_{\cN(S)}^{-1/2}\circ \cN(\phi)\right).
\end{align*}
In particular, if we choose $\ket\phi=S^{-1/2} \sqtrho \ket \psi$ for some $\ket\psi\in \supp(P)$, then, it follows from \eqref{eq:normalization} that $\ket\phi$ is normalized, and 
\begin{align*}
F^2(\phi, \mathcal D\circ \cN(\phi))= \tr\left(\cN(\sqtrho \psi \sqtrho) \,\cdot\, \Gamma_{\cN(S)}^{-1/2} \circ \cN\circ \Gamma_S^{-1/2}(\sqtrho \psi \sqtrho)\right).
\end{align*}

\paragraph{Flip operators:} The flip operator (cf.~\Cref{def:flip-operator} and \Cref{lem:flip-operator}) of the code space is given by
\begin{align*}
F_S &=  \sum_{i,j=1}^m \ket{w_i}\bra{w_j}\otimes \ket{w_j}\bra{w_i} \\
& = \left(S^{-1/2}\sqtrho\otimes S^{-1/2} \sqtrho\right) \left(  \sumi_{i,j=1}^m \ket{v_i}\bra{v_j}\otimes \ket{v_j} \bra{v_i} \right) \left(\sqtrho S^{-1/2}\otimes \sqtrho S^{-1/2}\right)\\
& = \left(S^{-1/2}\sqtrho\otimes S^{-1/2} \sqtrho\right)  F_P \left(\sqtrho S^{-1/2}\otimes \sqtrho S^{-1/2}\right).
\end{align*}

\paragraph{Average fidelity of the code:}
We are now ready to compute the average fidelity of the code $\supp(S)$:
\begin{align*}
 \cFavg(\Pi_S)& \coloneqq \Favg (\cN; \supp(S)) \\
 & = \int_{\ket \phi\in \supp(S)} \dd \mu(\phi) F^2(\phi, \mathcal D\circ \cN(\phi)) \\
& = \int_{\ket \phi\in \supp(S)} \dd \mu(\phi) \tr \left(\cN\circ\Gamma_S^{1/2}(\phi) \cdot \Gamma_{\cN(S)}^{-1/2}\circ \cN(\phi)\right)\\
& = \int_{\ket \phi\in \supp(S)} \dd \mu(\phi) \tr \left( F_B \, \cdot \,\cN\circ\Gamma_S^{1/2}(\phi) \otimes \Gamma_{\cN(S)}^{-1/2}\circ \cN(\phi)\right),
\end{align*}
where $\mu(\cdot)$ denotes the Haar measure, and $F_B$ is the flip operator corresponding to $\mathcal H_B$. Using the definition of the adjoint map and \Cref{lem:flip-operator}, we continue:
\begin{align*}
\cFavg(\Pi_S) & =\int_{\ket \phi\in \supp(S)} \dd \mu(\phi) \tr \left(
\left(\Gamma_{S}^{1/2}\circ \cN^* \otimes \cN^{*}\circ \Gamma_{\cN(S)}^{-1/2}\right)(F_B) \cdot \phi\otimes \phi
\right)\\
& = \frac{1}{m(m+1)}  \tr \left(
\left(\Gamma_{S}^{1/2}\circ \cN^* \otimes \cN^{*}\circ \Gamma_{\cN(S)}^{-1/2}\right)(F_B) \cdot (\Pi_S\otimes \Pi_S + F_S)
\right)\\
& = \frac{1}{m(m+1)}  \tr \left(
\left(\Gamma_{S}^{1/2}\circ \cN^* \otimes \cN^{*}\circ \Gamma_{\cN(S)}^{-1/2}\right)(F_B) \cdot \left(\Gamma_{S^{-1/2}\sqtrho}\otimes \Gamma_{S^{-1/2}\sqtrho}\right)(P\otimes P + F_P)
\right)\\
& = \int_{\ket\psi\in \supp(P)} \dd \mu(\psi)  \tr \left(
\left(\Gamma_{S}^{1/2}\circ \cN^* \otimes \cN^{*}\circ \Gamma_{\cN(S)}^{-1/2}\right)(F_B) \right.\\
&\qquad\qquad \left. \cdot \left(\Gamma_{S^{-1/2}\sqtrho}\otimes \Gamma_{S^{-1/2}\sqtrho}\right)(\psi\otimes \psi) \right)\\
& = \int_{\ket\psi\in \supp(P)} \dd \mu(\psi)  \tr \left(
F_B \cdot \left(    \cN\circ \Gamma_S^{1/2} \circ\Gamma_{S^{-1/2}\sqtrho} \right)(\psi) \otimes  \left(           \Gamma_{\cN(S)}^{-1/2}      \circ\cN\circ \Gamma_{S^{-1/2}\sqtrho}\right)( \psi) \right)\\
& = \int_{\ket\psi\in \supp(P)} \dd \mu(\psi)  \tr \left(\left(    \cN\circ \Gamma_S^{1/2} \circ\Gamma_{S^{-1/2}\sqtrho} \right)(\psi) \,\cdot\,  \left(           \Gamma_{\cN(S)}^{-1/2}      \circ\cN\circ \Gamma_{S^{-1/2}\sqtrho}\right)( \psi) \right)\\
& = \int_{\ket\psi\in \supp(P)} \dd \mu(\psi)  \tr \left(\left(    \cN\circ \Gamma_{\sqtrho} \right)(\psi) \,\cdot\,  \left(           \Gamma_{\cN(S)}^{-1/2}      \circ\cN\circ \Gamma_{S^{-1/2}\sqtrho}\right)( \psi) \right),
\end{align*}
where in the last line we used $S^{1/2}S^{-1/2}\sqtrho\ket\psi = \Pi_S\sqtrho\ket\psi = \sqtrho\ket\psi$ for every $\ket\psi\in \supp(P)$.

\paragraph{Average fidelity in terms of collision relative entropy:} We now express the average fidelity in terms of the collision relative entropy.
\begin{align*}
&\tr \left(\left(    \cN\circ \Gamma_{\sqtrho} \right)(\psi) \,\cdot\,\left(  \Gamma_{\cN(S)}^{-1/2}      \circ\cN  \circ \Gamma_{S^{-1/2}\sqtrho}\right)( \psi) \right) \\
 &  \qqquad  =   \tr \left(\left(    \cN\circ \Gamma_{\sqtrho} \right)(\psi) \,\cdot\,  \left(           \Gamma_{\cN(S)}^{-1/2}      \circ\cN\circ \Gamma_{\sqtrho}\right)( \psi) \right) & \\
 & \qqquad   \qquad +  \tr \left(\left(    \cN\circ \Gamma_{\sqtrho} \right)(\psi) \,\cdot\,  \left(           \Gamma_{\cN(S)}^{-1/2}      \circ\cN \circ(\Gamma_{S^{-1/2}\sqtrho}  - \Gamma_{\sqtrho})\right)( \psi) \right)\\
  & \qqquad   =   \tr \left(\cN(\sqtrho\psi\sqtrho) \,\cdot\, \cN(S)^{-1/2}    \,\cdot\,  \cN(\sqtrho \psi\sqtrho)\,\cdot\, \cN(S)^{-1/2} \right) & \\
 & \qqquad   \qquad +  \tr \left(\left(  \Gamma_{\cN(S)}^{-1/2}\circ  \cN\circ \Gamma_{\sqtrho} \right)(\psi) \,\cdot\,  \cN\left(     S^{-1/2}\sqtrho \psi \sqtrho S^{-1/2}  -   \sqtrho \psi \sqtrho          \right) \right)\\
 & \qqquad   = \exp D_2\left( \cN(\sqtrho \psi \sqtrho ) \| \cN(S)  \right) \\
 & \qqquad   \qquad +  \tr \left(\left(  \Gamma_{\cN(S)}^{-1/2}\circ  \cN\circ \Gamma_{\sqtrho} \right)(\psi) \,\cdot\,  \cN\left(     S^{-1/2}\sqtrho \psi \sqtrho S^{-1/2}  -   \sqtrho \psi \sqtrho          \right) \right) \\
 & \qqquad   \geq  \exp D_2\left( \cN(\sqtrho \psi \sqtrho ) \| \cN(S)  \right) \\
 & \qqquad   \qquad -   \left\|\left(  \Gamma_{\cN(S)}^{-1/2}\circ  \cN\circ \Gamma_{\sqtrho} \right)(\psi) \,\cdot\,  \cN\left(     S^{-1/2}\sqtrho \psi \sqtrho S^{-1/2}  -   \sqtrho \psi \sqtrho          \right) \right\|_1
\end{align*}
Therefore, 
\begin{multline*}
\cFavg(\Pi_S)  \geq   \int_{\ket\psi\in \supp(P)} \dd \mu(\psi)  \exp D_2\left( \cN(\sqtrho \psi \sqtrho ) \| \cN(S)  \right) \\
  -   \int_{\ket\psi\in \supp(P)} \dd \mu(\psi)  \left\|\left(  \Gamma_{\cN(S)}^{-1/2}\circ  \cN\circ \Gamma_{\sqtrho} \right)(\psi) \,\cdot\,  \cN\left(     S^{-1/2}\sqtrho \psi \sqtrho S^{-1/2}  -   \sqtrho \psi \sqtrho          \right) \right\|_1.
\end{multline*}

\paragraph{Random code space:} We now assume that the projection $P$ (of rank $m$)
is chosen randomly with respect to the Haar measure. Then the code space $\supp(S)$ itself becomes random. Hence, we can bound the expectation of the average fidelity of our code by
\begin{align*}
\int_P \dd & \mu(P)    \cFavg(\Pi_S) \\
& \geq  \int_P \dd \mu(P) \int_{\ket\psi\in \supp(P)} \dd \mu(\psi)  \exp D_2\left( \cN(\sqtrho \psi \sqtrho ) \| \cN(S)  \right) \\
 & \qquad -   \int_P \dd \mu(P) \int_{\ket\psi} \dd \mu(\psi)  \left\|\left(  \Gamma_{\cN(S)}^{-1/2}\circ  \cN\circ \Gamma_{\sqtrho} \right)(\psi) \,\cdot\,  \cN\left(     S^{-1/2}\sqtrho \psi \sqtrho S^{-1/2}  -   \sqtrho \psi \sqtrho          \right) \right\|_1\\
 & \eqqcolon T_1 - T_2,\numberthis\label{eq:two-terms}
\end{align*}
where the integration in the third line is again over 
$|\psi\rangle\in\supp(P)$. 
We refer to the terms appearing in the second and third line as $T_1$ and $T_2$ respectively. To analyse the first term $T_1$, we use \Cref{lem:key-lemma}. 
Note that the term $T_2$ vanishes if $\rho$ is a maximally mixed state (and correspondingly $\ket{\Psi^\rho}$ is a maximally entangled state).

\paragraph{Analysis of the first term in \eqref{eq:two-terms}:} The first step in analysing $T_1$ in \eqref{eq:two-terms} is to use Fubini's theorem to change the order of integrals:
\begin{align*}
T_1   = \int_P \dd \mu(P) \int_{\ket\psi\in \supp(P)} \dd \mu(\psi)  \exp D_2\left( \cN(\sqtrho \psi \sqtrho ) \| \cN(\sqtrho P\sqtrho)  \right).
\end{align*}
In the above integral the state $\ket \psi\in \supp(P)$ is distributed according to the Haar measure. Furthermore, for any such state $\ket \psi$ there is a projection $P'$ of rank $m-1$ such that $P=\psi + P'$. Indeed, if we let 
$$\ket{\psi}^{\perp} \coloneqq\{\ket v:\, \bra \psi  v\rangle =0, \ket v\in \cH_A\},$$
be the orthogonal subspace of $\ket \psi$, then $P'$ is a projection onto some $(m-1)$-dimensional subspace of $\ket\psi^{\perp}.$ Moreover, this projection $P'$ is distributed, independent of $\ket \psi$, according to the Haar measure. Putting these together, we find that
\begin{align*}
T_1  = \int_{\psi} \dd\mu(\psi)\int_{P': \supp(P')\subseteq \ket\psi^{\perp}} \dd\mu(P')\exp D_2\left( \cN(\sqtrho \psi \sqtrho ) \| \cN(\sqtrho \psi\sqtrho) +\cN(\sqtrho P'\sqtrho)  \right).
\end{align*}

The next step is to use the joint convexity of the exponential of the collision relative entropy from \Cref{lem:collision-convex}. We obtain 
\begin{align*}
T_1  & \geq  \int_{\psi} \dd\mu(\psi)\exp D_2\left( \cN(\sqtrho \psi \sqtrho ) \middle\| \cN(\sqtrho \psi\sqtrho) +\int_{P': \supp(P')\subseteq \ket\psi^{\perp}} \dd\mu(P') \cN(\sqtrho P'\sqtrho)  \right)\\
& = \int_{\psi} \dd\mu(\psi)\exp D_2\left( \cN(\sqtrho \psi \sqtrho ) \middle\| \cN(\sqtrho \psi\sqtrho) + \frac{m-1}{d-1} \cN\left(\sqtrho (I-\psi)\sqtrho\right)  \right),
\end{align*}
where in the second line we use
\begin{align*}
\int_{P': \supp(P')\subseteq \ket\psi^{\perp}} \dd\mu(P') P' = \frac{m-1}{d-1}(I-\psi).
\end{align*}

Then, using $\tilde \rho = d\rho$ the above bound can be written as
\begin{align*}
T_1  & \geq \int_{\psi} \dd \mu(\psi) \exp D_2(\cN(\sqtrho \psi \sqtrho ) \| \alpha \cN(\sqtrho \psi \sqtrho ) + \beta \cN(\rho )   ),
\end{align*}
where 
\begin{align}\label{eq:alpha-beta}
\alpha=\frac{d-m}{d-1}\qquad\text{and}\qquad \beta= \frac{d(m-1)}{d-1}.
\end{align}
Note that $\alpha+\beta = m$. 

To choose the random (Haar distributed) vector $\ket \psi$, we may first choose a random (Haar distributed) unitary operator $U$ and then take $\ket \psi=\ket {u_i^*}$, for a fixed $i$, where $\ket{u_i}$ is given by \Cref{def:dephasing-map}. We then have
\begin{align*}
T_1  & \geq \int_{U} \dd\mu(U)  \frac{1}{d} \sum_{i=1}^d\exp D_2\left(\cN(\sqtrho u_i^* \sqtrho ) \middle\| \alpha \cN(\sqtrho u_i^* \sqtrho ) + \beta \cN(\rho )   \right)\\
& = \int_{U} \dd\mu(U)  \frac{1}{d} \sum_{i=1}^d\exp D_2\left(u_i\otimes \cN(\sqtrho u_i^* \sqtrho ) \middle\| \alpha u_i\otimes \cN(\sqtrho u_i^* \sqtrho ) + \beta u_i\otimes \cN(\rho )   \right)\\
& \geq \int_{U} \dd\mu(U)  \exp D_2\left(  \frac{1}{d} \sum_{i=1}^d u_i\otimes  \cN(\sqtrho u_i^* \sqtrho ) \middle\| \alpha \frac{1}{d} \sum_{i=1}^d u_i\otimes \cN(\sqtrho u_i^* \sqtrho ) + \beta \pi_R\otimes \cN(\rho )   \right),
\end{align*}
where the last inequality follows from the joint convexity of $\exp D_2(\cdot \| \cdot)$.

We now use \Cref{lem:dephasing-map} and \Cref{lem:key-lemma} to obtain:
\begin{align*}
T_1& \geq \int_U \dd \mu(U) \exp D_2\left(\mathcal T_U \otimes (\cN\circ \Gamma_{\sqtrho})(\Phi) \middle\| \alpha \mathcal T_U \otimes (\cN\circ \Gamma_{\sqtrho})(\Phi) + \beta \pi_R\otimes \cN(\rho)\right)\\
& \geq  \int_U \dd \mu(U) \frac 1 d \exp D_2\left(\id_R \otimes (\cN\circ \Gamma_{\sqtrho})(\Phi) \middle\| \alpha \mathcal T_U \otimes (\cN\circ \Gamma_{\sqtrho})(\Phi) + \beta \pi_R\otimes \cN(\rho)\right).
\end{align*}
Note that here we applied \Cref{lem:key-lemma} with the choice $\Lambda = \cN\circ\Gamma_{\sqtrho}$ and 
\begin{align*}
\sigma_{RB} = \alpha \cT_U\otimes (\cN\circ\Gamma_{\sqtrho})(\Phi) +\beta \pi_R\otimes \cN(\rho),
\end{align*}
which satisfies $\Gamma_{Z_U}\ox \id_B(\sigma_{RB})=\sigma_{RB}$.

Once again using the joint convexity of $\exp D_2(\cdot\| \cdot)$, we find that 
\begin{align*}
T_1 & \geq \frac 1 d \exp D_2\left(\id_R \otimes (\cN\circ \Gamma_{\sqtrho})(\Phi) \middle\| \alpha\int_U \dd\mu(U) \mathcal T_U \otimes (\cN\circ \Gamma_{\sqtrho})(\Phi) + \beta \pi_R\otimes \cN(\rho)\right).
\end{align*}
Using \Cref{lem:dephasing-map}, we compute 
\begin{align*}
\int_{U} \dd\mu(U)\, \mathcal T_U\otimes \id_A(\Phi) &= \int_{U} \dd\mu(U)\, \frac 1 d\sum_{i=1}^d u_i\otimes u_i^*\\
&  = \int_{\psi} \dd\mu(\psi)  \psi\otimes \psi^*\\
& = x \Phi + y \pi_R\otimes \pi_A,
\end{align*}
where 
\begin{align*}
x= \frac{1}{d+1}\qquad \text{and} \qquad y = \frac{d}{d+1}.
\end{align*}
Therefore,
\begin{align*}
T_1& \geq  \frac 1 d\exp D_2\left(\id_R \otimes (\cN\circ \Gamma_{\sqtrho})(\Phi) \middle\| (\alpha x)\id_R\otimes (\cN\circ \Gamma_{\sqtrho})(\Phi) + (\alpha y+\beta) \pi_R\otimes \cN(\rho)\right)\\
& =  \frac{1}{md}\exp D_2\left(\id_R \otimes (\cN\circ \Gamma_{\sqtrho})(\Phi) \middle\| \alpha'\id_R\otimes (\cN\circ \Gamma_{\sqtrho})(\Phi) + \beta' \pi_R\otimes \cN(\rho)\right)\\
& =  \frac{1}{md}\exp D_2\left(\id_R \otimes \cN(\Psi^{\rho}) \middle\| \alpha'\id_R\otimes \cN(\Psi^{\rho}) + \beta' \pi_R\otimes \cN(\rho)\right),
\end{align*}
where we used $\Psi^\rho = \id_R\otimes \Gamma_{\sqtrho}(\Phi)$ and 
\begin{align*}
\alpha' &= \frac{1}{m}\alpha x = \frac{(d-m)}{m(d-1)(d+1)},\\
\beta' &= \frac{1}{m}(\alpha y+ \beta) = \frac{(d-m)d}{m(d-1)(d+1)}+ \frac{d(m-1)}{m(d-1)}.
\end{align*}
Observe that $\alpha'+\beta' = 1$, which follows from \eqref{eq:alpha-beta}. 

Finally, we use \Cref{lem:collision-spectrum} to bound the collision relative entropy by the information spectrum relative entropy. 
For any $\delta_1\in (0, 1)$ we obtain
\begin{align*}
T_1\geq \frac{1}{md}(1-\delta_1)\left[\alpha' + \beta'\exp\left(-D_s^{\delta_1}(\id_R\otimes \cN(\Psi^\rho)\| \pi_R\otimes \cN(\rho))\right)\right]^{-1}.
\end{align*}
Let us fix $\eps_1\in (\delta_1, 1)$. Then $T_1\geq 1-\eps_1$ follows from assuming that 
\begin{align*}
\frac{1}{md}(1-\delta_1)\left[\alpha' + \beta'\exp\left(-D_s^{\delta_1}(\id_R\otimes \cN(\Psi^\rho)\| \pi_R\otimes \cN(\rho))\right)\right]^{-1}\geq 1-\eps_1.
\end{align*}
This inequality is equivalent to
\begin{align*}
\exp\left(-D_s^{\delta_1}(\id_R\otimes \cN(\Psi^{\rho})\| \pi_R\otimes \cN(\rho))\right) \leq \frac{1-\delta_1}{md \beta'(1-\eps_1)} -\frac{\alpha'}{\beta'}.
\end{align*}
Using the facts that $\beta'\leq 1$ and $\alpha'/\beta'\leq 1/(md)$, this inequality holds if 
\begin{align*}
\exp\left(-D_s^{\delta_1}(\id_R\otimes \cN(\Psi^{\rho})\| \pi_R\otimes \cN(\rho))\right) \leq \frac{1-\delta_1}{md(1-\eps_1)} -\frac{1}{md}= \frac{\eps_1-\delta_1}{md(1-\eps_1)},
\end{align*}
which is equivalent to 
\begin{align*}
\log m & \leq D_s^{\delta_1}(\id_R\otimes \cN(\Psi^\rho)\| \pi_R\otimes \cN(\rho)) -\log d + \log\frac{\eps_1-\delta_1}{1-\eps_1}\\
& =  D_s^{\delta_1}(\id_R\otimes \cN(\Psi^\rho)\| I_R\otimes \cN(\rho)) + \log\frac{\eps_1-\delta_1}{1-\eps_1}.
\end{align*}
In summary, we obtain $T_1\geq 1-\eps_1$ for $\eps_1>0$, provided that for some $0<\delta_1<\eps_1$ we have
\begin{align}\label{eq:Cg1}
\log m \leq D_s^{\delta_1}(\id_R\otimes \cN(\Psi^\rho)\| I_R\otimes \cN(\rho)) + \log\frac{\eps_1-\delta_1}{1-\eps_1}.
\end{align}

\paragraph{Analysis of the second term in \eqref{eq:two-terms}:} The main ideas in this part have already appeared in the analysis of the first term of \eqref{eq:two-terms} above, so we leave the details for \Cref{app:analysis-t2}.
There, assuming 
\begin{align*}
\log m \leq D_s^{\delta_2} (\Psi^\rho\| I_R\otimes \rho ) + \log\frac{\eps_2 - \delta_2}{1- \eps_2}
\end{align*}
for $0<\delta_2<\eps_2$, we show that 
\begin{align*}
T_2^2 & \leq 1+\tr(\rho^2)-  (1-\eps_2)\leq \tr(\rho^2) + \eps_2.
\end{align*}

\paragraph{The last step:} 
Putting the above bounds on $T_1$ and $T_2$ together, we conclude that
\begin{align*}
\int_P \dd  \mu(P)    \cFavg(\Pi_S) 
 & \geq T_1 - T_2
 \geq 1- \eps_1 - \sqrt{\tr(\rho^2) + \eps_2}
\end{align*}
holds if 
\begin{align*}
\log m \leq \min\left\lbrace   D_s^{\delta_1}(\id_R\otimes \cN(\Psi^\rho)\| I_R\otimes \cN(\rho)) + \log\frac{\eps_1-\delta_1}{1-\eps_1},  D_s^{\delta_2} (\Psi^\rho\| I_R\otimes \rho ) + \log\frac{\eps_2 - \delta_2}{1- \eps_2} \right\rbrace.
\end{align*}
Then, using \Cref{lem:ave-ent-fid} and writing the entanglement fidelity in terms of the average fidelity, we conclude that there exists a projection operator $P$ such that, for the code space $\supp(S)$, where $S= \sqrt{\trho} P \sqrt{\trho}$, we have
\begin{align*}
\Fent(\cN; \supp(S)) \geq 1-\eps,
\end{align*}
with $\eps$ defined by 
\begin{align*}
\eps = \left( \eps_1 + \sqrt{\tr(\rho^2)+\eps_2}\right)\left(1+\frac{1}{d}\right).
\end{align*}
This means that 
\begin{multline*}
 \Qet{1}{\eps}(\cN)\geq \log  \min \left\{  \left\lfloor \exp\left( D_s^{\delta_1}(\id_R\otimes \cN(\Psi^\rho)\|  I_R\otimes \cN(\rho)) + \log\frac{\eps_1-\delta_1}{1-\eps_1}\right)\right\rfloor, \right.\\
\left. \left\lfloor \exp\left(D_s^{\delta_2} (\Psi^\rho\| I_R\otimes \rho ) + \log\frac{\eps_2 - \delta_2}{1- \eps_2}\right)\right\rfloor \right\},
\end{multline*}
which gives the desired bound in~\Cref{thm:et-one-shot}. The fact that $ \Q{1}{\eps}(\cN)\geq  \Qet{1}{\eps}(\cN)$ is already stablished in~\eqref{eq:Q-Qet-compare}.

\subsection{Asymptotic expansion}\label{sec:et-asymptotic}

We now prove the main theorem of this section:

\begin{theorem}[Second order achievability bound]\label{thm:et-second-order}
Let $\cN\colon A \rightarrow B$ be a quantum channel and fix $\eps\in (0,1)$. Then
we have
\begin{align}\label{asym-et}
\Q{n}{\eps}(\cN)\geq \Qet{n}{\eps}(\cN) \geq n I_c(\cN) + \sqrt{nV_\eps(\cN)}\,\invP{\eps} + \cO(\log n),
\end{align}
where the coherent information $I_c(\cN)$ and the $\eps$-quantum dispersion $V_\eps(\cN)$ are given by \Cref{def:channel-coherent-information}.
\end{theorem}

\begin{proof}
In the following we assume that $(\id_R\otimes \cN)(\Psi^\rho)$ is not a pure state, since otherwise the channel $\cN$ is an isometry for which the problem of quantum information (or entanglement) transmission is trivial, and the claimed achievability bound is immediate. 
In other words, letting $\cU_\cN\colon \cH_A\rightarrow \cH_{BE}$ be a Stinespring isometry of $\cN$, we assume that in 
\begin{align*}
\omega_{RBE} = (\id_R\otimes \mathcal \cU_{\cN})(\Psi^\rho)
\end{align*}
 the environment $E$ is \emph{not} completely decoupled from $R$. This implies that $I(R;E)_{\omega}>0$, which is equivalent to 
\begin{align}\label{eq:lower-bound-H-rho}
H(B)_\omega-H(RB)_\omega < H(R)_\omega=H(\rho).
\end{align}

By \Cref{thm:et-one-shot} we have the following bound on $\Qet{n}{\eps}(\cN)$, for an arbitrary input state $\rho_n\in\cD(\cH_A^{\otimes n})$:
\begin{multline}\label{eq:et-n-blocklength-bound}
\Qet{n}{\eps}(\cN) \geq \min\left\lbrace D_s^{\delta_1}(\id_{R^n}\ox \cN(\Psi^{\rho_n})\|I_{R^n}\ox \cN(\rho_n)) + \log \frac{\eps_1-\delta_1}{1-\eps_1},\right.\\
\left.  D_s^{\delta_2}(\Psi^{\rho_n}\| I_{R^n}\ox \rho_n) + \log \frac{\eps_2-\delta_2}{1-\eps_2} \right\rbrace
\end{multline}
where $\eps_i,\delta_i>0$ for $i=1,2$ are chosen such that 
\begin{align*}
\eps = \left( \eps_1 + \sqrt{\tr(\rho^2)+\eps_2}\right)\left(1+\frac{1}{d}\right),
\end{align*}
with $d=\dim\cH_A$ and $0\leq \delta_i\leq \eps_i$. In the following, we restrict our consideration to input states of the form $\rho^{\otimes n}$ where $\rho\in\cD(\cH_A)$.

Fix $\eps > 0$, and for sufficiently large $n$ define $\eps_1>0$ by 
$$\eps_1 =\eps - \frac{3}{ \sqrt n}.$$
Furthermore, define
$$\eps_2 = \left(   \eps\left(1+ \frac{1}{d^n}\right)^{-1} - \eps_1   \right)^2      - \tr\left( (\rho^{\otimes n})^2  \right),$$
such that \eqref{eq:def-tilde-eps} holds.
Then for sufficiently large $n$, we have
\begin{align*}
\eps_2 & =  \left(   \frac{3}{\sqrt n} +\eps\left( \left(1+ \frac{1}{d^n}\right)^{-1} - 1\right)   \right)^2      - \tr\left( (\rho^{\otimes n})^2  \right)\\
& \geq  \left(  \frac{3}{ \sqrt n} -\frac{1}{1+d^n}\eps    \right)^2      - \tr\left( (\rho^{\otimes n})^2  \right) \\
&\geq \frac{9}{n} - \frac{3\eps}{\sqrt{n}(1+d^n)} - \tr\left( (\rho^{\otimes n})^2  \right)\\
&\geq \frac{4}{n}.
\end{align*}
Note that here we assume that $\rho$ is not pure (since otherwise there is nothing to prove), so that $\tr\big( (\rho^{\otimes n})^2  \big)$ tends to zero exponentially fast in $n$. 
Finally, let 
\begin{align}\label{eq:eps-delta}
\delta_1= \eps_1- \frac{1}{\sqrt n } = \eps -\frac{4}{\sqrt n}\qquad \text{and}\qquad \delta_2 = \eps_2 - \frac{1}{ n}\geq \frac{3}{{n}}.
\end{align}
Then by \Cref{prop:entropies-SOA} and \Cref{lem:phi-trick}, we obtain the following expansion for the first term in \eqref{eq:et-n-blocklength-bound}:
\begin{multline}
D_s^{\delta_1}(\id_{R^n}\otimes \cN^{\otimes n}\big((\Psi^\rho)^{\otimes n}\big)\|  I_{R^n}\otimes \cN^{\otimes n}(\rho^{\otimes n})) + \log\frac{\eps_1-\delta_1}{1-\eps_1}\\
 = n D(\id_R\otimes \cN(\Psi^\rho)\| I_R\otimes \cN(\rho)  ) + \sqrt{n V} \Phi^{-1}(\eps) + \Theta(\log n),\label{eq:coh-a-b}
\end{multline}
where $V= V(\id_R\otimes \cN(\Psi^\rho)\| I_R\otimes \cN(\rho)  )$.

In \Cref{app:proof-4-10} we show that 
\begin{align}
D_s^{\delta_2} \big((\Psi^\rho)^{\otimes n}\| I_{R^n}\otimes \rho^{\otimes n} \big) + \log\frac{\eps_2 - \delta_2}{1- \eps_2} \geq n H(\rho) - O(\log n\sqrt n).
\label{eq:56}
\end{align}
Observe that
$$D(\id_R\otimes \cN(\Psi^\rho)\| I_R\otimes \cN(\rho)  ) = H(B)_\omega - H(RB)_{\omega},$$
where $\omega_{RBE} = (\id_R\otimes \mathcal \cU_{\cN})(\Psi^\rho)$. Then by \eqref{eq:lower-bound-H-rho} we obtain
$$D(\id_R\otimes \cN(\Psi^\rho)\| I_R\otimes \cN(\rho)  ) = H(B)_\omega - H(RB)_{\omega} < H(R)_\omega = H(\rho).$$
Therefore, for sufficiently large $n$, the expression \eqref{eq:coh-a-b} is less than \eqref{eq:56}, and hence the lower bound in \eqref{eq:et-n-blocklength-bound} is given in terms of the first term. That is, for every $\eps >0$ and arbitrary $\rho\in\cD(\cH_A)$ we have
\begin{align}\label{eq:et-arbitrary-input-bound}
\Qet{n}{\eps}(\cN) \geq n D(\id_R\otimes \cN(\Psi^\rho)\| I_R\otimes \cN(\rho)  ) + \sqrt{n V(\id_R\otimes \cN(\Psi^\rho)\| I_R\otimes \cN(\rho)  )} \Phi^{-1}(\eps) + O(\log n).
\end{align}
Since \eqref{eq:et-arbitrary-input-bound} holds for any arbitrary input state $\rho\in\cD(\cH)$, we have
\begin{align*}
\Qet{n}{\eps}(\cN) &\geq \max_{\rho\in\cD(\cH)} \left\lbrace n D(\id_R\otimes \cN(\Psi^\rho)\| I_R\otimes \cN(\rho)  ) + \sqrt{n V(\id_R\otimes \cN(\Psi^\rho)\| I_R\otimes \cN(\rho)  )} \Phi^{-1}(\eps) \right\rbrace\\
&\qquad  + O(\log n).\\
&\geq \max_{\rho\in\cS_c(\cN)} \left\lbrace n D(\id_R\otimes \cN(\Psi^\rho)\| I_R\otimes \cN(\rho)  ) + \sqrt{n V(\id_R\otimes \cN(\Psi^\rho)\| I_R\otimes \cN(\rho)  )} \Phi^{-1}(\eps) \right\rbrace\\
&\qquad  + O(\log n),
\end{align*}
where $\cS_c(\cN)$ is the set of states defined in \eqref{eq:Pi_c}. We have
\begin{align*}
D(\id_R\otimes \cN(\Psi^\rho)\| I_R\otimes \cN(\rho) = I_c(\cN)
\end{align*}
for all states $\rho\in\cS_c(\cN)$, and noting that $\invP{\eps}<0$ for $\eps<1/2$ (resp.~$\invP{\eps}\geq 0$ for $\eps\geq 1/2$), we obtain
\begin{align*}
\Qet{n}{\eps}(\cN) \geq \begin{cases}
n I_c(\cN) + \sqrt{n \min_{\rho\in\cS_c(\cN)}V(\cN,\rho)}\,\invP{\eps} + \cO(\log n) & \text{if }\eps\in (0,1/2)\\
n I_c(\cN) + \sqrt{n \max_{\rho\in\cS_c(\cN)}V(\cN,\rho)}\,\invP{\eps} + \cO(\log n) & \text{if }\eps\in (1/2,1)
\end{cases} 
\end{align*}
where $V(\cN,\rho) = V(\id_R\otimes \cN(\Psi^\rho)\| I_R\otimes \cN(\rho)  )$. The proof is then completed by employing definition \eqref{v-eps} of $V_\eps(\cN)$.
\end{proof}

\section{Example: 50-50 erasure channel}\label{sec:example}

Let $\cN\colon A\rightarrow B$ be the 50-50 (symmetric) erasure channel, which has zero capacity \cite{BDS97} by the No-cloning theorem. In this section, we study the $\eps$-error $n$-blocklength capacities of $\cN$, based on ideas from~\cite{MW14}. 

Let $\eps\in (0, 1/2)$ and $m=\exp \Qet{n}{\eps}(\cN)$. Then there is a code $(m, \cH_M, \cD)$ with 
\begin{align}\label{eq:erasure-f-eps}
\Fent(\cN^{\otimes n}; \cH_M) =  F^2(\Phi_{R^nA^n}^m, (\id_{R^n}\otimes \cD\circ \cN^{\otimes n})(\Phi_{R^nA^n}^{m}))\geq 1-\eps.
\end{align}
Furthermore, let $\ket{\omega_{R^nB^{n}E^{n}}}$ be a purification of $(\id_{R^n}\otimes \cN^{\otimes n})(\Phi^m_{R^nA^n})$. Note that 
since $\cN$ is a symmetric channel, $\ket{\omega_{R^nB^nE^n}}$ can be chosen to be symmetric with respect to the exchange of the subsystems $B^n$ and $E^n$. 

It is easy to verify that for any state $\sigma_{A^n}$ we have 
\begin{align}\label{eq:fidelity-bound}
F^2(\Phi^m_{R^nA^n}, I_{R^n}\otimes \sigma_{A^n}) = \langle \Phi_{R^nA^n}^m|\one_{R^n}\ox \sigma_{A^n}|\Phi_{R^nA^n}^m\rangle \leq \frac{1}{m}.
\end{align}
On the other hand, by \eqref{eq:erasure-f-eps} we have $\Phi_{R^nA^n}^m\in \cB_{\sqrt \eps}(\theta_{R^nA^n})$ where $\theta_{R^nA^n}\coloneqq(\id_{R^n}\otimes \cD) (\omega_{R^nB^n}).$
Therefore, 
\begin{align*}
H_{\max}^{\sqrt \eps}(R^n|A^n)_\theta & =\min_{\bar{\theta}\in \cB_{\sqrt \eps}(\theta)} \max_{\sigma_{A^n}} \log F^2\left(\bar{\theta}_{R^nA^n}, I_{R^n}\otimes \sigma_{A^n}\right)\\
& \leq \max_{\sigma_{A^n}} \log F^2(\Phi^{m}_{R^nA^n}, I_{R^n}\otimes \sigma_{A^n})\\
& \leq -\log m,
\end{align*}
where we used \eqref{eq:fidelity-bound} in the last line.
We continue to bound:
\begin{align*}
\Qet{n}{\eps}(\cN) & = \log m\\
&\leq - H_{\max}^{\sqrt \eps}(R^n|A^n)_\theta\\
& \leq -H_{\max}^{\sqrt \eps}(R^n|B^n)_{\omega}\\
& = H_{\min}^{\sqrt \eps}(R^n|E^n)_{\omega}\\
& = H_{\min}^{\sqrt \eps}(R^n|B^n)_{\omega}\\
& \leq H_{\max}^{\sqrt \eps}(R^n|B^n)_{\omega} + \log\frac{1}{\cos^2(2\alpha)},\numberthis\label{eq:chain-3-6}
\end{align*}
where $\alpha$ is chosen such that $\sin\alpha=\sqrt \eps$. Here, in the third line we used the data processing inequality for the smooth max-entropy, \Cref{lem:Hmin-Hmax}(iii).  In the fourth line we used the duality relation for the smooth min- and max-entropies, \Cref{lem:Hmin-Hmax}(i). In the fifth line we used the symmetry of $\omega_{R^nB^nE^n}$, and in the last line we used \Cref{lem:Hmin-Hmax}(ii) as well as the fact that $\eps<1/2$. The latter implies that $\alpha$, defined through $\sin\alpha=\sqrt \eps$, satisfies $\alpha\in (0, \pi/4)$.

Comparing the third and last lines of \eqref{eq:chain-3-6}, we find that 
$$\Qet{n}{\eps}(\cN)\leq -H_{\max}^{\sqrt \eps}(R^n|B^n)_{\omega} \leq\log\frac{1}{\cos(2\alpha)}.$$
We conclude that we have a constant upper bound on $\Qet{n}{\eps}(\cN)$ for arbitrary $n$ and
$\eps<1/2$.

On the other hand, consider the case $\eps\in (1/2, 1)$. By \Cref{thm:et-second-order}, we have
\begin{align*}
\Qet{n}{\eps}(\cN) \geq n I_c(\cN) + \sqrt{nV_\eps(\cN)}\,\invP{\eps} + \cO(\log n)
\end{align*}
where $V_\eps(\cN) = \max_{\rho\in\cS_c(\cN)}V(\omega_{RB}\|I_R\ox\omega_B)$ for this range of $\eps$ (cf.~\Cref{def:channel-coherent-information}) and $\omega_{RB} = (\id_R\ox\cN)(\Phi_{RA}^m)$. A straightforward calculation verifies that $I_c(\cN, \rho)\leq 0$ for all $\rho$, and that $I_c(\cN)=I_c(\cN,\pi_A)=0$. On the other hand, by considering the maximally mixed state as the input state, we have $V_\eps(\cN)>0$ for $\eps>1/2$. It follows that
\begin{align*}
\Qet{n}{\eps}(\cN) \geq \sqrt{nV_\eps(\cN)}\,\invP{\eps} + \cO(\log n),
\end{align*}
and the right-hand side of the above inequality is positive for sufficiently large $n$ since $\invP{\eps}>0$ for $\eps>1/2$. 

To summarize, for $\eps<1/2$, the $\eps$-error $n$-blocklength capacity $\Qet{n}{\eps}(\cN)$ is at most a constant independent of $n$, whereas $\Qet{n}{\eps}(\cN)$ is positive and scales as $\sqrt n$ for $\eps>1/2$.


\section*{Acknowledgements}
\addcontentsline{toc}{section}{Acknowledgements}
We would like to thank Francesco Buscemi and Amin Gohari for useful discussions, and Mark Wilde for helpful feedback and pointing out an error in a previous version of the paper.

\appendix
\section{Properties of distance measures and entropic quantities}\label{sec:props-entropies}
In this appendix we collect useful properties of the distance measures and entropic quantities defined in \Cref{sec:entropies}.

In~\cite{HHH99} (see also~\cite{Nie02}) the following relationship between the average fidelity and the entanglement fidelity was proven:
\begin{lemma} {\normalfont~\cite{HHH99} } \label{lem:ave-ent-fid}
For any quantum operation $\Lambda$ acting on $\cB(\cH)$ with $d=\dim \cH$, the average fidelity and entanglement fidelity are related by
\begin{align*}
\Favg(\Lambda;\cH)  = \frac{d\, \Fent(\Lambda; \cH) + 1}{d + 1}.
\end{align*}
\end{lemma}
The relative entropy and the quantum information variance satisfy the following duality relations:
\begin{lemma}\label{lem:duality}
Let $\ket{\psi_{ABC}}$ be a pure state with corresponding marginals $\rho_{AB},\rho_{AC},$ and $\rho_{BC}$, then
\begin{align*}
D(\rho_{AB}\|\one_A\ox\rho_{B}) = - D(\rho_{AC}\|\one_A\ox\rho_{C}) \qquad\text{and}\qquad V(\rho_{AB}\|\one_A\ox\rho_{B}) = V(\rho_{AC}\|\one_A\ox\rho_{C}).
\end{align*}
\end{lemma}
These relations have been used in~\cite{HayashiTomamichel14}. Here, we give a proof for the sake of completeness.
\begin{proof}
The relation $D(\rho_{AB}\| I_A\otimes \rho_B)=-D(\rho_{AC}\| I_A\otimes \rho_C) $ follows from a straightforward calculation. For the second equation, we only need to establish that 
\begin{align*}
	\tr \left[\rho_{AB}(\log \rho_{AB}- I_A\otimes \log \rho_B)^2\right] = \tr \left[\rho_{AC}(\log \rho_{AC} - I_A\otimes \log\rho_C)^2\right],
\end{align*}
or equivalently, that
\begin{align*}
\tr \left[\rho_{AB} \log^2 \rho_{AB}\right] &+ \tr \left[\rho_B\log^2\rho_B\right] - 2\tr \left[\rho_{AB}(\log \rho_{AB}) (I_A\otimes \log \rho_B)\right]\\
& = \tr \left[\rho_{AC} \log^2 \rho_{AC}\right] + \tr \left[\rho_C\log^2\rho_C\right] - 2\tr \left[\rho_{AC}(\log \rho_{AC}) (I_A\otimes \log \rho_C)\right].
\end{align*}
Considering the Schmidt decompositions of $\ket{\psi_{ABC}}$ along the cuts $AB/C$ and $AC/B$, one verifies that $\tr \left[\rho_{AB} \log^2 \rho_{AB}\right] = \tr \left[\rho_C\log^2\rho_C\right]$ and $\tr \left[\rho_B\log^2\rho_B\right]=\tr \left[\rho_{AC} \log^2 \rho_{AC}\right]$. It remains to be shown that 
$$\tr \left[\rho_{AB}(\log \rho_{AB}) (I_A\otimes \log \rho_B)\right] = \tr \left[\rho_{AC}(\log \rho_{AC}) (I_A\otimes \log \rho_C)\right],$$
which is equivalent to
$$\bra{ \psi_{ABC}}  (\log \rho_{AB}\otimes I_C) (I_{AC}\otimes \log \rho_B)\ket {\psi_{ABC}} = \bra{ \psi_{ABC}}  (I_{AB}\otimes \log \rho_{C}) (\log \rho_{AC}\otimes I_B)\ket {\psi_{ABC}}.$$
Once again using Schmidt decomposition, we find that $\log \rho_{AB}\otimes I_C\ket {\psi_{ABC}} = I_{AB}\otimes \log \rho_C\ket {\psi_{ABC}}$ and that $I_{AC}\otimes \log \rho_B\ket {\psi_{ABC}} =  \log \rho_{AC}\otimes I_B\ket {\psi_{ABC}}$, which concludes the proof.
\end{proof}

The next lemmas concern the smooth min- and max-entropies.
\begin{lemma}\label{lem:Hmin-Hmax}
The following properties hold:
\begin{enumerate}[{\normalfont (i)}]
\item {\normalfont \cite{KRS09}} Let $\rho_{ABC}$ be a pure state and $\eps\in(0,1)$, then
\begin{align*}
\Hmin^\eps(A|B)_\rho = -\Hmax^\eps(A|C)_\rho.
\end{align*}
\item {\normalfont \cite{Tom12}} Let $\rho_{AB}\in\cD(\cH_{AB})$ and $\alpha, \beta>0$ be such that $\alpha+\beta<\pi/2$. Then we have
\begin{align*}
\Hmin^{\sin \alpha}(A|B)_\rho \leq \Hmax^{\sin\beta}(A|B)_\rho + \log\frac{1}{\cos^2(\alpha+\beta)}.
\end{align*}
In particular, if $\eps,\eps'>0$ such that $\eps+\eps'<1$, then 
\begin{align*}
\Hmin^\eps(A|B)_\rho \leq \Hmax^{\eps'}(A|B)_\rho + \log\frac{1}{1-(\eps+\eps')^2}.
\end{align*}
\item {\normalfont \cite{TCR10}} Data-processing inequality for max-entropy: Let $\rho_{AB}\in\cD(\cH_{AB})$, $\eps\in (0,1)$, and $\Lambda\colon B\rightarrow D$ be a CPTP map with $\tau_{AD}\coloneqq (\id_A\ox\Lambda)(\rho_{AB})$, then	
\begin{align*}
\Hmax^\eps(A|B)_\rho \leq \Hmax^\eps(A|D)_\tau.
\end{align*}
\end{enumerate}
\end{lemma}
\begin{lemma}\label{lem:Hmin-first-order}
{\normalfont \cite{TCR09, Tom12}}
Let $\rho\in\cD(\cH_{AB})$ and $\eps\in (0,1)$, then
\begin{align*}
\Hmin^\eps(A^n|B^n)_{\rho^\n} \geq n H(A|B)_\rho - \sqrt{n}\cO\left(\sqrt{g(\eps)}\right)
\end{align*}
where $g(t) = - \log(1-\sqrt{1-t^2})$.
\end{lemma}

We use the following convexity property of the collision relative entropy.
\begin{lemma}\label{lem:collision-convex}
{\normalfont \cite{MDS+13, WWY14}} The function $(\rho, \sigma)\mapsto \exp D_2(\rho\|\sigma)$ is jointly convex. 
\end{lemma}

The following inequality between the collision relative entropy and the information spectrum relative entropy is used in our proof of the main result.

\begin{lemma}\label{lem:collision-spectrum}
{\normalfont \cite{BG13}} 
For every $0< \eps, \lambda <1$, $\rho\in\cD(\cH)$ and $\sigma\in\cP(\cH)$ we have
\begin{align*}\exp D_2&(\rho\| \lambda \rho  + (1-\lambda)\sigma) \geq
 (1-\eps)\left[\lambda + (1-\lambda) \exp\left(-D_s^{\eps}(\rho \| \sigma) \right)\right]^{-1}.
\end{align*}
\end{lemma}

The information spectrum entropy and the smooth max-relative entropy (cf.~\Cref{def:max-relative-entropy}) can be bounded by each other:
\begin{lemma}\label{lem:D_s-D_H-D_max}
{\normalfont \cite{TH13}} Let $\rho,\sigma\in\cD(\cH)$, $\eps\in (0,1)$, and $\eta>0$. Furthermore, let $\nu(\sigma)$ denote the number of different eigenvalues of $\sigma$. Then we have
\begin{align*}
\Dmax^{\sqrt{1-\eps}}(\rho\|\sigma) \leq D_s^{\eps+\eta}(\rho\|\sigma) - \log\eta + \log\nu(\sigma) - \log(1-\eps).
\end{align*}
\end{lemma}
\begin{remark}
\Cref{lem:D_s-D_H-D_max} can be obtained by combining two results in \cite{TH13} which bound the information spectrum entropy and the smooth max-relative entropy in terms of the \emph{hypothesis testing relative entropy} \cite{WR12}. These bounds are proved in Lemma 12 and Proposition 13 of \cite{TH13}, respectively.
\end{remark}

The second order asymptotic expansion of the smooth max-relative entropy and the information spectrum relative entropy are derived in~\cite{TH13}. While the former is one of the main results therein, the latter is only proved implicitly. 
\begin{proposition}\label{prop:entropies-SOA}
{\normalfont \cite{TH13}} Let $\eps\in(0,1)$ and $\rho,\sigma\in\cD(\cH)$, then we have the following second order asymptotic expansions for the smooth max-entropy and the information spectrum relative entropy:
\begin{align*}
\Dmax^\eps(\rho^\n\|\sigma^\n) &= n D(\rho\|\sigma) - \sqrt{n V(\rho\|\sigma)}\,\invP{\eps^2} + \cO(\log n)\\
D_s^\eps(\rho^\n\|\sigma^\n) &= n D(\rho\|\sigma) + \sqrt{n V(\rho\|\sigma)}\,\invP{\eps} + \cO(\log n)
\end{align*}
\end{proposition}

\section{The flip operator}

In the proof of \Cref{thm:et-one-shot} in \Cref{sec:thm-4.1-proof} we make use of the following operator and its properties:

\begin{definition}[Flip operator on a subspace]\label{def:flip-operator}
Let $\cK\subseteq \cH$ be a subspace of the Hilbert space $\cH$. We define the \emph{flip operator} $F_\cK$ as the linear extension of the operator defined by the action 
\begin{align*}
F_\cK(|\psi_1\rangle\ox|\psi_2\rangle) = |\psi_2\rangle \ox |\psi_1\rangle
\end{align*}
for every $|\psi_1\rangle,|\psi_2\rangle\in \cK$. For a Hermitian operator $X\in\cB(\cH)$, we define $F_X \equiv F_{\supp(X)}$.
\end{definition}

The proofs of the following properties of the flip operator can be found e.g.~in~\cite{DBWR14}.

\begin{lemma} {\normalfont \cite{DBWR14}} \label{lem:flip-operator}
Let $\cK\subseteq \cH$ be a subspace of the Hilbert space $\cH$ with $\dim\cK = d_\cK$, then the following properties hold:
\begin{enumerate}[{\normalfont (i)}]
\item Given an orthonormal basis $\lbrace |v_i\rangle\rbrace_{i=1}^{d_\cK}$ of $\cK$, the flip operator can be expressed as
\begin{align*}
F_\cK = \sum_{i,j=1}^{d_\cK} |v_i\rangle\langle v_j|\ox |v_j\rangle \langle v_i|.
\end{align*}
\item For operators $X,Y$ acting on $\cK$, we have $\tr(XY) = \tr(F_\cK(X\ox Y))$.
\item The flip operator is idempotent on its support, i.e.,~$F_\cK^2 = \Pi_\cK\ox\Pi_\cK$ where $\Pi_\cK$ is the orthogonal projection onto $\cK$.
\item We have:
\begin{align*}
\int_{\ket\psi\in \cK}\diff\mu(\psi) \psi\ox\psi = \frac{1}{d_\cK(d_\cK+1)}(\Pi_\cK\ox\Pi_\cK + F_\cK)
\end{align*}

\item More generally,
$$ \int_P\dd\mu(P) P\otimes P=  \gamma_1 \Pi_\cK\otimes \Pi_\cK + \gamma_2 F_\cK$$
where the integral is taken over rank-$m$ orthogonal projections with respect to the Haar measure and
\begin{align*}
\gamma_1=\frac{m(md_\cK-1)}{d_\cK(d_\cK^2-1)}   \qquad \qquad \gamma_2= \frac{m(d_\cK-m)}{d_\cK(d_\cK^2-1)}.
\end{align*}

\end{enumerate}
\end{lemma}

\section{Proof of \texorpdfstring{\Cref{thm:et-one-shot}}{Theorem \ref{thm:et-one-shot}}: Analysis of the second term in (\ref{eq:two-terms})}\label{app:analysis-t2}
Recall that
$$T_2 = \int_P \dd \mu(P) \int_{\ket\psi\in \supp(P)} \dd \mu(\psi)  \Big\|\big(  \Gamma_{\cN(S)}^{-1/2}\circ  \cN\circ \Gamma_{\sqtrho} \big)(\psi) \,\cdot\,  \cN\big(     S^{-1/2}\sqtrho \psi \sqtrho S^{-1/2}  -   \sqtrho \psi \sqtrho          \big) \Big\|_1.$$
Let us denote the 1-norm expression under the integral by $T'_2(P, \psi)$, and observe that 
\begin{align*}
T'_2(P, \psi) 
& \leq \left\|\left(  \Gamma_{\cN(S)}^{-1/2}\circ  \cN\circ \Gamma_{\sqtrho} \right)(\psi)\right\|_\infty \cdot \left\|\cN\left(     S^{-1/2}\sqtrho \psi \sqtrho S^{-1/2}  -   \sqtrho \psi \sqtrho          \right) \right\|_1.
\end{align*}
Since $\ket\psi\in \supp(P)$, we have $\psi\leq P$ and 
\begin{align*}
0&\leq \left(  \Gamma_{\cN(S)}^{-1/2}\circ  \cN\circ \Gamma_{\sqtrho} \right)(\psi)\\
& = \Gamma_{\cN(S)}^{-1/2} \circ \cN(\sqtrho \psi \sqtrho)\\
& \leq \Gamma_{\cN(S)}^{-1/2} \circ\cN(\sqtrho P \sqtrho)\\
& = \Gamma_{\cN(S)}^{-1/2} \circ\cN(S)\\
& \leq I_B.
\end{align*}
Therefore, 
\begin{align*}
T'_2(P, \psi) 
& \leq \left\|\cN\left(     S^{-1/2}\sqtrho \psi \sqtrho S^{-1/2}  -   \sqtrho \psi \sqtrho          \right) \right\|_1\\
& \leq \left\|     S^{-1/2}\sqtrho \psi \sqtrho S^{-1/2}  -   \sqtrho \psi \sqtrho          \right\|_1,
\end{align*}
where in the last line we used the fact that $\cN$ is a CPTP map.

Let $\ket\phi = S^{-1/2}\sqtrho \ket\psi$ and $\ket{\phi'}=\sqtrho\ket\psi$. Then, as verified in \eqref{eq:normalization}, the vector $\ket\phi$ is normalized. Thus, a straightforward calculation yields
\begin{align*}
\left\|     S^{-1/2}\sqtrho \psi \sqtrho S^{-1/2}  -   \sqtrho \psi \sqtrho          \right\|_1 & = \|\phi - \phi'\|_1\\
& = \left[1+ \left\|\ket{\phi'}\right\|^4 - \left\|\ket{\phi'}\right\|^2  - |\langle \phi | \phi'\rangle |^2\right]^{1/2}\\
& = \left[1+ \bra\psi \trho \ket \psi^2 - \bra \psi\trho\ket\psi - \bra\psi \sqtrho S^{-1/2}\sqtrho\ket\psi^2\right]^{1/2}.
\end{align*}
Then, using the concavity of the square root function we have 
\begin{multline*}
\int_{\ket\psi\in \supp(P)}\dd \mu(\psi) T'_2(P, \psi) \\ \leq \left[ \int_{\ket\psi\in \supp(P)} \dd\mu(\psi)  \left(1+ \bra\psi \trho \ket \psi^2 - \bra \psi\trho\ket\psi - \bra\psi \sqtrho S^{-1/2}\sqtrho\ket\psi^2 \right) \right]^{1/2}.
\end{multline*}
We compute each term under the integral individually. 
The easiest one is the third one:
$$\int_{\ket\psi\in \supp(P)} \dd\mu(\psi)  \bra \psi\trho\ket\psi =  \int_{\ket\psi\in \supp(P)} \dd\mu(\psi) \tr(\trho \psi)= \frac{1}{m}\tr(\trho P) = \frac{d}{m}\tr(\rho P)$$
For the second term, we compute using \Cref{lem:flip-operator}:
\begin{align*}
\int_{\ket\psi\in \supp(P)} \dd\mu(\psi)  \bra \psi\trho\ket\psi^2 & = \int_{\ket\psi\in \supp(P)} \dd\mu(\psi) \tr(\trho \psi \trho \psi)\\
& = \int_{\ket\psi\in \supp(P)} \dd\mu(\psi) \tr(F_P (\trho \psi \otimes \trho \psi ))\\
& = \int_{\ket\psi\in \supp(P)} \dd\mu(\psi) \tr(F_P (\trho\otimes \trho) (\psi\otimes \psi) )\\
& = \frac{1}{m(m+1)} \tr(F_P (\trho\otimes \trho) (P\otimes P + F_P))\\
& = \frac{1}{m(m+1)} \left(   \tr\left( (\trho P)^2\right) + (\tr(\trho P))^2    \right)
\end{align*}
We express the last term as
\begin{align*}
\bra\psi \sqtrho S^{-1/2}\sqtrho\ket\psi^2 & = \bra \psi\sqtrho S^{-1/2} \sqtrho \ket \psi \bra \psi \sqtrho S^{-1/2}\sqtrho \ket \psi \\
& =  \tr\left(   S^{-1/2} \sqtrho  \psi  \sqtrho S^{-1/2}\sqtrho  \psi   \sqtrho \right)\\
& = \exp D_2\left(\sqtrho \psi \sqtrho \| S\right).
\end{align*}
Putting these together, we find that
\begin{multline*}
\int_{\ket\psi\in \supp(P)}\dd \mu(\psi) T'_2(P, \psi) \leq \left[  1+ \frac{1}{m(m+1)} \left(   \tr\left( (\trho P)^2\right) + (\tr(\trho P))^2    \right)\right. \\
 \left. - \frac{d}{m}\tr(\rho P) - \int_{\ket\psi\in \supp(P)} \dd\mu(\psi)  \exp D_2\left(\sqtrho \psi \sqtrho \| S\right)    \right]^{1/2}.
\end{multline*}
Again using the concavity of the square root function, we obtain
\begin{align*}
T_2^2 & = \left[\int_P\dd\mu(P) \int_{\ket\psi\in \supp(P)}\dd \mu(\psi) T'_2(P, \psi)\right]^2 \\
& \leq  1+ \frac{1}{m(m+1)} \int_P\dd\mu(P)   \tr\left( (\trho P)^2\right) + \frac{1}{m(m+1)} \int_P\dd\mu(P)(\tr(\trho P))^2   \\
& \qquad  - \frac{d}{m}  \int_P\dd\mu(P)\tr(\rho P) -  \int_P\dd\mu(P)\int_{\ket\psi\in \supp(P)} \dd\mu(\psi)  \exp D_2\left(\sqtrho \psi \sqtrho \| S\right).   
\end{align*}

Observe that 
$$\int_{P} \dd\mu(P) P = \frac{m}{d} I\qquad \text{and}\qquad \int_P\dd\mu(P) P\otimes P=  \gamma I_A\otimes I_A + \kappa F_A$$
where
$$ \gamma=\frac{m(md-1)}{d(d^2-1)}   \qquad \text{and}\qquad \kappa= \frac{m(d-m)}{d(d^2-1)}.$$
We note that $\gamma d^2+ \kappa d= \tr(P\otimes P) = m^2$.
Therefore,
$$\frac{d}{m}  \int_P\dd\mu(P)\tr(\rho P) = 1,$$
and
\begin{align*}
\int_P\dd\mu(P)   \tr\left( (\trho P)^2\right)& = \int_P\dd\mu(P)  \tr(F_A (\trho\otimes \trho) (P\otimes P)) \\
& = \gamma\tr(F_A(\trho\otimes \trho)) + \kappa \tr(F_A(\trho\otimes \trho) F_A)\\
& = \gamma\tr(\trho^2) + \kappa\tr(\trho)^2\\
& = \gamma d^2 \tr(\rho^2) + \kappa d^2.
\end{align*}
We similarly have
\begin{align*}
\int_P\dd\mu(P)(\tr(\trho P))^2 & = \int_P\dd\mu(P) \tr((\trho\otimes \trho) (P\otimes P))\\
& = \gamma\tr(\trho\otimes \trho) + \kappa \tr((\trho\otimes \trho) F_A)\\
&= \gamma d^2  + \kappa d^2 \tr\left(\rho^2\right).
\end{align*}
Putting these together, we arrive at
\begin{align*}
T_2^2 & \leq \frac{1}{m(m+1)}  (\gamma d^2 + \kappa d^2)(1+\tr(\rho^2) )-  \int_P\dd\mu(P)\int_{\ket\psi\in \supp(P)} \dd\mu(\psi)  \exp D_2\left(\sqtrho \psi \sqtrho \| S\right)\\
& \leq 1+\tr(\rho^2)-  \int_P\dd\mu(P)\int_{\ket\psi\in \supp(P)} \dd\mu(\psi)  \exp D_2\left(\sqtrho \psi \sqtrho \| S\right),
\end{align*}
where in the second line we used $\gamma d^2 \leq m^2$ and $\kappa d^2\leq m$. To get an upper bound on this expression, we need a lower bound on the last term with the double integral. For this we repeat the same process as before: we first change the order of integral, then write the result in terms of dephasing maps, and finally use \Cref{lem:key-lemma}. We first have
\begin{align*}
\int_P\dd\mu(P)\int_{\ket\psi\in \supp(P)} \dd\mu(\psi)  & \exp D_2\left(\sqtrho \psi \sqtrho \| S\right)\\
& = \int_P\dd\mu(P)\int_{\ket\psi\in \supp(P)} \dd\mu(\psi)  \exp D_2\left(\sqtrho \psi \sqtrho \| \sqtrho P\sqtrho \right)\\
& = \int_{\psi} \dd\mu(\psi) \int_{P':\supp(P')\subseteq \ket\psi^{\perp}}\dd\mu(P') \exp D_2\left(\sqtrho \psi \sqtrho \| \sqtrho P\sqtrho \right)\\
&  \geq \int_{\psi} \dd \psi       \exp D_2  \left(\sqtrho \psi \sqtrho \| \alpha \sqtrho \psi \sqtrho + \beta \rho \right),
\end{align*}
where in the last line we used the joint convexity of $\exp D_2(\cdot \|\cdot)$. Next,  writing the result in terms of dephasing maps and using \Cref{lem:key-lemma} as before, we arrive at
\begin{align*}
\int_P\dd\mu(P)\int_{\ket\psi\in \supp(P)} \dd\mu(\psi)   \exp D_2\left(\sqtrho \psi \sqtrho \middle\| S\right) \geq \frac{1}{md} \exp D_2 \left(\Psi^\rho \| \alpha' \Psi^\rho + \beta' \pi_R\otimes \rho \right).
\end{align*}
Finally, using \Cref{lem:collision-spectrum} we find that for any $\delta_2\in(0, 1)$ we have
\begin{multline*}
\int_P\dd\mu(P)\int_{\ket\psi\in \supp(P)} \dd\mu(\psi)  \exp D_2\left(\sqtrho \psi \sqtrho \| S\right)\\
 \geq  \frac{1}{md} (1-\delta_2) \left[ \alpha'  + \beta' \exp\left(-D_s^{\delta_2}(\Psi^\rho \|    \pi_R\otimes \rho)  \right)   \right]^{-1}.
\end{multline*}
We now assume that 
$$\frac{1}{md} (1-\delta_2) \left[ \alpha'  + \beta' \exp\left(-D_s^{\delta_2}(\Psi^\rho \|    \pi_R\otimes \rho)  \right)   \right]^{-1} \geq 1-\eps_2$$
for some $\eps_2>0$. Repeating the same calculations as before, this inequality holds if 
$$\log m \leq D_s^{\delta_2} (\Psi^\rho\| I_R\otimes \rho ) + \log\frac{\eps_2 - \delta_2}{1- \eps_2}.$$
Then assuming the above inequality, we obtain 
\begin{align*}
T_2^2 & \leq 1+\tr(\rho^2)-  (1-\eps_2)\leq \tr(\rho^2) + \eps_2.
\end{align*}


\section{Proof of \texorpdfstring{\Cref{eq:56}}{Equation (\ref{eq:56})}} \label{app:proof-4-10}
Observe that 
$$\log\frac{\eps_2-\delta_2}{1-\eps_2} \geq -\log n,$$
and that
$$D_s^{\delta_2} \left((\Psi^\rho)^{\otimes n}\| I_{R^n}\otimes \rho^{\otimes n} \right) \geq D_s^{\frac {3} n} \left((\Psi^\rho)^{\otimes n}\| I_{R^n}\otimes \rho^{\otimes n} \right),$$
since $\delta_2\geq 3/n$ by \eqref{eq:eps-delta}.
 Then it suffices to show that 
\begin{align}\label{eq:56-2}
D_s^{\frac {3} n} \left((\Psi^\rho)^{\otimes n}\| I_{R^n}\otimes \rho^{\otimes n} \right)\geq  n D(\Psi^\rho\| I_R\otimes \rho) - O(\log n\sqrt n).
\end{align}

Using \Cref{lem:D_s-D_H-D_max} we find that 
\begin{align*}
D_s^{\frac {3} n} \left((\Psi^\rho)^{\otimes n}\| I_{R^n}\otimes \rho^{\otimes n} \right) \geq D_{\max}^{\sqrt{1-2/n}}\left((\Psi^\rho)^{\otimes n}\| I_{R^n}\otimes \rho^{\otimes n} \right) - \log \nu(I_{R^n}\otimes \rho^{\otimes n}) + \log(1/n- 2/n^2)
\end{align*}
where $\nu(I_{R^n}\otimes \rho^{\otimes n})$ is the number of different eigenvalues of $I_{R^n}\otimes \rho^{\otimes n}$. We note that $\nu(I_{R^n}\otimes \rho^{\otimes n})$ grows polynomially in $n$. 

Let $\mathcal B= \cB_{\sqrt{1-2/n}}\left((\Psi^{\rho})^{\otimes n}\right)$. We continue to bound:
\begin{align*}
D_{\max}^{\sqrt{1-2/n}}\left((\Psi^\rho)^{\otimes n}\| I_{R^n}\otimes \rho^{\otimes n} \right)& = \min_{\omega_n\in \mathcal B} D_{\max}(\omega_n \| I_{R^n}\otimes \rho^{\otimes n})\\
& \geq \min_{\sigma_n} \min_{\omega_n\in \mathcal B} D_{\max}(\omega_n \| I_{R^n}\otimes \sigma_n)\\
& = \min_{\omega_n\in \mathcal B} \min_{\sigma_n} D_{\max}(\omega_n \| I_{R^n}\otimes \sigma_n)\\
& = - \max_{\omega_n\in \mathcal B}\left[- \min_{\sigma_n} D_{\max}(\omega_n \| I_{R^n}\otimes \sigma_n)\right]\\
& = - H_{\min}^{\sqrt{1- 2/n}} (A^n|R^n)_{(\Psi^{\rho})^{\otimes n}}\\
& = H_{\max}^{\sqrt{1-2/n}}(A^n)_{\rho^{\otimes n}},
\end{align*}
where in the last line we used the duality of the min- and max-entropies, \Cref{lem:Hmin-Hmax}(i).

Next, using \Cref{lem:Hmin-Hmax}(ii), we have
$$ H_{\max}^{\sqrt{1-2/n}}(A^n)_{\rho^{\otimes n}} \geq H_{\min}^{1/n^2}(A^n)_{\rho^{\otimes n}} + \log\left(1-\left(\sqrt{1-\frac 2 n} + \frac{1}{n^2}\right)^2\right).$$
Observe that the following holds for sufficiently large $n$:
\begin{align*}
1-\left(\sqrt{1-\frac 2 n} + \frac{1}{n^2}\right)^2 &= 1- \left( 1-\frac 2 n + \frac{1}{n^4}   + \frac{1}{n^2}\sqrt{1- \frac 2 n}  \right)\\
&= \frac 2 n - \frac{1}{n^4}   - \frac{1}{n^2}\sqrt{1- \frac 2 n}  \\
& \geq \frac 1 n.
\end{align*}
Therefore,
$$ H_{\max}^{\sqrt{1-2/n}}(A^n)_{\rho^{\otimes n}} \geq H_{\min}^{1/n^2}(A^n)_{\rho^{\otimes n}} - O(\log n).$$
Finally, using \Cref{lem:Hmin-first-order} we have
$$H_{\min}^{1/n^2}(A^n)_{\rho^{\otimes n}} \geq nH(A)_\rho - O\left(\sqrt{g(1/n^2)}\right)\sqrt n,$$
where $g(t)=- \log (1-\sqrt{1-t^2}).$
Observe that 
$\sqrt{1-t^2}\leq 1- t^2/2$. Thus, $-\log(1-\sqrt{1-t^2}) \leq -\log t^2/2$, and hence
\begin{align*}
H_{\min}^{1/n^2}(A^n)_{\rho^{\otimes n}} \geq nH(A)_\rho - O(\log n)\sqrt n.
\end{align*}
Putting all the above inequalities together yields \eqref{eq:56}.



\begin{thebibliography}{10}


\bibitem{BK02} H.~Barnum,~and E.~Knill, ``Reversing quantum dynamics with near-optimal quantum and classical fidelity," J.~Math.~Phys. \textbf{43}(5) 2097-2106 (2002).

\bibitem{BKN00} H.~Barnum, E.~Knill, and M.~A.~Nielsen, ``On quantum fidelities and channel capacities," IEEE Trans.~Inf.~Theory \textbf{46}(4), 1317-1329 (2000).

\bibitem{BG13} S.~Beigi and A.~Gohari, ``Quantum Achievability Proof via Collision Relative Entropy," IEEE Trans.~Inf.~Theory \textbf{60}(12), 7980-7986 (2014).

\bibitem{BDS97} C.~H.~Bennett, D.~P.~DiVincenzo, and J.~A.~Smolin, ``Capacities of Quantum Erasure Channels," Phys. Rev. Lett. \textbf{78}, 3217-3220 (1997).

\bibitem{BBC+93} C.~H.~Bennett, G.~Brassard, C.~Cr\'epeau, R.~Jozsa, A.~Peres, and W.~K.~Wootters,
 ``Teleporting an unknown quantum state via dual classical and Einstein-Podolsky-Rosen channels," Phys.~Rev.~Lett.~\textbf{70}(13), 1895 (1993).
 
 
\bibitem{BD10} F.~Buscemi and N.~Datta, ``Distilling entanglement from arbitrary resources," J.~Math.~Phys. \textbf{51}(10), 102201 (2010).
 
\bibitem{BD10a} F.~Buscemi and N.~Datta, ``The quantum capacity of channels with arbitrarily correlated noise," IEEE Trans.~Inf.~Theory \textbf{56}(3), 1447-1460 (2010).
 
\bibitem{Dat09} N.~Datta, ``Min- and Max-relative entropies and a new entanglement monotone," IEEE Trans.~Inf.~Theory \textbf{55}(6), 2816-2826  (2009).
 
\bibitem{DH11} N.~Datta and M.-H.~Hsieh, ``The apex of the family tree of protocols: Optimal rates and resource inequalities," New J.~Phys.~\textbf{13}(9) 093042  (2011).
 
\bibitem{DL14b} N.~Datta and F.~Leditzky, ``Second-Order Asymptotics for Source Coding, Dense Coding, and Pure-State Entanglement Conversions," IEEE Trans.~Inf.~Theory \textbf{61}(1), 582-608 (2015).
 
\bibitem{Dev05} I.~Devetak, ``The private classical capacity and quantum capacity of a quantum channel,"
IEEE Trans.~Inf.~Theory \textbf{51}(1), 44-55 (2005).
 
\bibitem{DS05} I.~Devetak and P.~W.~Shor, ``The capacity of a quantum channel for simultaneous transmission of classical and quantum information," Commun.~Math.~Phys. \textbf{256}(2), 287-303 (2005).
 
\bibitem{DBWR14} F.~Dupuis, M.~Berta, J.~Wullschleger, and R.~Renner, ``One-Shot Decoupling," Commun.~Math.~Phys. \textbf{328}(1), 251-284  (2014).
 
\bibitem{FR14} O.~Fawzi and R.~Renner, ``Quantum conditional mutual information and approximate Markov chains," Commun.~Math.~Phys. \textbf{340}(2), 575-611 (2015).

\bibitem{Hay09} M.~Hayashi, ``Information spectrum approach to second-order coding rate in channel coding," IEEE Trans.~ Inf.~Theory \textbf{55}(11), 4947-4966  (2009).

\bibitem{Hay08} M.~Hayashi, ``Second-order asymptotics in fixed-length source coding and intrinsic randomness," IEEE Trans.~Inf.~Theory \textbf{54}(10), 4619-4637 (2008).

\bibitem{HayashiTomamichel14} M.~Hayashi and M.~Tomamichel, ``Correlation Detection and an Operational Interpretation of the R\'enyi Mutual Information," 2015 IEEE International Symposium on Information Theory
(ISIT), 1447-1451 (2015).

\bibitem{HSW08} P.~Hayden, P.~W.~Shor, and A.~Winter, ``Random quantum codes from Gaussian ensembles
and an uncertainty relation," Open Syst.~Inf.~Dyn. \textbf{15}(1), 71-89  (2008).

\bibitem{HHWY08} P.~Hayden, M.~Horodecki, A.~Winter, and J.~Yard, ``A decoupling approach to the quantum
capacity'' Open Syst.~Inf.~Dyn. \textbf{15}(1), 7-19 (2008).

\bibitem{HJPW04} P.~Hayden, R.~Jozsa, D.~Petz, and A.~Winter, ``Structure of states which satisfy strong subadditivity of quantum entropy with equality,'' Commun.~Math.~Phys. \textbf{246}(2), 359-374  (2004).

\bibitem{HHH99} M.~Horodecki, P.~Horodecki, and R.~Horodecki, ``General teleportation channel, singlet fraction, and quasidistillation," Phys.~Rev.~A \textbf{60}(3), 1888 (1999).

\bibitem{JRS+15} M.~Junge, R.~Renner, D.~Sutter, M.~M.~Wilde, and A.~Winter, ``Universal recovery maps and approximate sufficiency of quantum relative entropy," Ann.~Henri Poincare \textbf{19}(10), 2955-2978 (2018).

\bibitem{KRS09} R.~K\"onig, R.~Renner, and C.~Schaffner, ``The operational meaning of min-and max-entropy," IEEE Trans.~Inf.~Theory \textbf{55}(9), 4337-4347  (2009).

\bibitem{KV13a} V.~Kostina and S.~Verd\'u, ``Nonasymptotic noisy lossy source coding," IEEE Information Theory Workshop (ITW), 1-5 (2013).

\bibitem{KW04} D.~Kretschmann and R.~F.~Werner, ``Tema con variazioni: quantum channel capacity," New J.~Phys. \textbf{6}(1), 26 (2004).

\bibitem{KH13} W.~Kumagai and M.~Hayashi, ``Second order asymptotics for random number generation," 2013 IEEE International Symposium on Information Theory Proceedings (ISIT), 1506-1510 (2013).

\bibitem{Li14} K.~Li, ``Second order asymptotics for quantum hypothesis testing," Ann.~Statistics \textbf{42}(1), 171-189 (2014).

\bibitem{Llo97} S.~Lloyd, ``Capacity of the noisy quantum channel," Phys.~Rev.~A \textbf{55}(3), 1613 (1997). 

\bibitem{MW14} C.~Morgan and A.~Winter, ````Pretty strong" converse for the quantum capacity of degradable
channels," IEEE Trans.~Inf.~Theory \textbf{60}(1), 317-333  (2014).

\bibitem{MDS+13} M.~M\"uller-Lennert, F.~Dupuis, O.~Szehr, S.~Fehr, and M.~Tomamichel, ``On quantum R\'enyi entropies: A new generalization and some properties," J.~Math.~Phys. \textbf{54}(12), 122203 (2013).

\bibitem{HKM10} H.~K.~Ng and P.~Mandayam, ``Simple approach to approximate quantum error correction
based on the transpose channel," Phys.~Rev.~A \textbf{81}(6), 062342  (2010).

\bibitem{Nie02} M.~A.~Nielsen, ``A simple formula for the average gate fidelity of a quantum dynamical
operation," Phys.~Lett.~A \textbf{303}(4), 249-252 (2002).

\bibitem{OP93} M.~Ohya and D.~Petz, \emph{Quantum Entropy and Its Use} (Springer, 1993).

\bibitem{Petz88} D.~Petz, ``Sufficiency of channels over von Neumann algebras," Quart.~J.~Math. \textbf{39}(1), 97-108 (1988).

\bibitem{Petz86} D.~Petz, ``Sufficient Subalgebras and the Relative Entropy of States of a von Neumann Algebra," Commun.~Math.~Phys. \textbf{105}(1), 123-131  (1986).

\bibitem{PPV10} Y.~Polyanskiy, V.~Poor, and S.~Verd\'u, ``Channel coding rate in the finite blocklength regime," IEEE Trans.~Inf.~Theory \textbf{56}(5), 2307-2359 (2010).

\bibitem{Ren05} R.~Renner, ``Security of quantum key distribution," PhD thesis, ETH Z\"urich (2005).

\bibitem{Sha48} C.~E.~Shannon, ``A mathematical theory of communication," Bell System Technical Journal \textbf{27}, 379-423 (1948).

\bibitem{Sho02} P.~W.~Shor, ``The quantum channel capacity and coherent information," Talk at MSRI Workshop on Quantum Computation. Berkeley, CA, USA (2002).

\bibitem{Str62} V.~Strassen, ``Asymptotische Absch\"atzungen in Shannons Informationstheorie," Transactions of the Third Prague Conference on Information Theory, 689-723 1962.

\bibitem{Tom12} M.~Tomamichel, ``A Framework for Non-Asymptotic Quantum Information Theory," PhD thesis, ETH Z\"urich (2012).

\bibitem{TBR15} M.~Tomamichel, M.~Berta, and J.~M.~Renes, ``Quantum coding with finite resources," Nature communications \textbf{7}, 11419 (2016).

\bibitem{TCR09} M.~Tomamichel, R.~Colbeck, and R.~Renner, ``A Fully Quantum Asymptotic Equipartition Property," IEEE Trans.~Inf.~Theory \textbf{55}(12), 5840-5847  (2009).

\bibitem{TCR10} M.~Tomamichel, R.~Colbeck, and R.~Renner, ``Duality between smooth min-and max-entropies," IEEE Trans.~Inf.~Theory \textbf{56}(9), 4674-4681 (2010).

\bibitem{TH13} M.~Tomamichel and M.~Hayashi, ``A hierarchy of information quantities for finite block length analysis of quantum tasks," IEEE Trans.~Inf.~Theory \textbf{59}(11), 7693-7710 (2013).

\bibitem{TWW14} M.~Tomamichel, M.~M.~Wilde, and A.~Winter, ``Strong Converse Rates for Quantum Communication," 2015 IEEE International Symposium on Information Theory (ISIT), 2386-2390 (2015).

\bibitem{WR12} L.~Wang and R.~Renner, ``One-shot classical-quantum capacity and hypothesis testing," Phys.~Rev.~Lett. \textbf{108}(20), 200501 (2012).

\bibitem{WWY14} M.~M.~Wilde, A.~Winter, and D.~Yang, ``Strong converse for the classical capacity of entanglement-breaking and Hadamard channels via a sandwiched R\'enyi relative entropy," Commun.~Math.~Phys. \textbf{331}(2), 593-622 (2014).



\end{thebibliography}
\end{document}